\documentclass[a4paper,11pt]{article}
\usepackage{jheppub} 
\usepackage{lineno}
\usepackage{slashed}
\usepackage{amsmath}     
\usepackage{amssymb}  
\usepackage{microtype}
\usepackage{tabularx}

\usepackage{graphicx}
\usepackage{booktabs}
\usepackage{array}
\usepackage{multirow}
\usepackage{slashed}
\usepackage{physics}
\usepackage{braket}
\newcommand{\openone}{\mathbb{I}}
\usepackage{booktabs}
\usepackage{ragged2e}

\usepackage{amsthm}
\theoremstyle{plain}
\newtheorem{prop}{Proposition}

\numberwithin{equation}{section}

\usepackage{physics}

\newtheorem{lemma}{Lemma}
\newtheorem{cor}{Corollary}
\newtheorem{remark}{Remark}

\usepackage{amsmath,amssymb,amsthm}
\theoremstyle{definition}

\title{\boldmath Spin-charge deconfinement from a self-consistent dressing of Fierz-complete $(1+1)$d Dirac fermions}

\author{L. Haddad, K. Gonzales, J. Kahanek, M. Kolding, J. Maguire, V. Palombi, and H. Truelson}
\affiliation{Department of Physics, Colorado School of Mines,\\
1523 Illinois Street, Golden, CO 80401, U.S.A.}

\emailAdd{lhaddad@mines.edu}

\abstract{We develop a self-consistent dressing for $(1+1)$d Dirac fermions with Fierz-complete four-fermion interactions, extending the spin-charge formalism of~\cite{Haddad2024}. By studying the associated composite connection, we characterize obstructions to trivializing the Dirac operator and prove a trivialization theorem under which spin-charge separation, half-infinite Wilson-line dressing, and flat-connection holonomy arise as regime limits of the single dressing. In particular, the chiral-difermion transition is a deconfinement transition for the spin and charge degrees of freedom, mediated by an emergent $\mathfrak{sl}(2,\mathbb{R})$ gauge field and measured by a Wilson-loop area law in the chiral phase. We show that our formulation provides a realization of the conjectured Faddeev--Niemi link between spin-charge separation and confinement. }

\begin{document}
\maketitle
\flushbottom

\section{Introduction}
\label{sec:introduction}

Spin-charge separation (SCS) is a feature of strongly correlated fermion systems in one dimension, originally identified for non-relativistic electrons in metals and developed within the context of Luttinger-liquid theory~\cite{Tomonaga1950,Luttinger1963,LutherEmery1974,Haldane1981,Giamarchi2004}. In the SCS regime, the electron ceases to exist as an elementary excitation and is replaced by independent spinon and chargon collective modes, the original fermion appearing only as a vertex operator built from bosonized phase fields. Whereas the standard derivation identifies SCS in the long-wavelength fluctuations around the linearized Fermi sea, in the present construction SCS arises in the fluctuations around the mean-field condensate of the difermion channel, with the bosonization derivation recovered as a special case (section~\ref{sec:spin-charge}). A related construction appears in gauge theory, i.e., the Dirac--Lavelle--McMullan dressing~\cite{Dirac1955,LavelleMcMullan1997,Mandelstam1968,Wilson1974,Schwinger1962}, in which a bare fermion is multiplied by a half-infinite Wilson line, producing a gauge-invariant physical charged state. These gauge-theoretic dressings have been systematized in the dressing field method~\cite{FournelFrancoisLazzariniMasson2014,FrancoisRavera2024JHEP,Ravera2026DFMLectures}, a general framework for constructing gauge-invariant composite fields out of connections; our aim is complementary, the extraction of emergent structure from a four-fermion theory rather than the reduction of a gauge symmetry. A third construction is the parallel transport of fermions in non-abelian backgrounds, organized by the holonomy of a principal-bundle connection and controlled on flat connections by the non-abelian Poincar\'e lemma~\cite{DemessieSaemann2014,HaydysIntroGauge}.

The three constructions that we treat, SCS, Wilson lines, and holonomy, look superficially different but their algebraic content is the same: a bare fermion field is multiplied by the exponential of an integrated one-form taking values in an appropriate dressing group. A construction is most useful when the underlying connection has a particular geometric structure (flatness) decoupling between commuting sectors, or integrability of the path-ordered exponential. The differences lie in the physical interpretation of the connection, the regime in which the construction is useful, and the technical tool used.

The identification of SCS with a gauge-theoretic confinement problem has a history in condensed matter physics. In the slave-particle formulation of the t-J model, spin and charge separate when the spinons and holons of a compact lattice gauge theory deconfine~\cite{NagaosaLee1990,LeeNagaosaWen2006}. Mudry and Fradkin showed that in one dimension the spinons of this gauge picture are always confined, and that the physical carriers of the separate spin and charge quantum numbers are solitons rather than the slave particles~\cite{MudryFradkin1994a,MudryFradkin1994b}. The construction developed in the present paper is consistent with that result: the binding field that we derive below is a non-compact connection built from the phase of a physical condensate, and the deconfined carriers in the paired phase are bosonic collective modes. This is consistent with the soliton resolution of~\cite{MudryFradkin1994a,MudryFradkin1994b} and with the vertex-operator structure of bosonization.

A highly developed attempt in this area is the program of Faddeev, Niemi, and collaborators in~\cite{NiemiWalet2005,ChernodubLandauGauge,FaddeevNiemi2007}, and in follow-up work~\cite{DittmannHeinzlWipf2001,vanBaalWipf2001,TsurumaruTsutsuiFujii2000} (which, in places, we will refer to collectively by \emph{CFN} for Cho-Faddeev-Niemi, or simply by \emph{Niemi}). The proposal is that SCS might extend to the low-energy domain of four-dimensional non-abelian gauge theory and play a role in color confinement. There, an SU(2) Yang--Mills gauge field is decomposed into a chiral spin field $\hat n(x)$ and a charged sector. The low-energy effective theory then takes the form of coupled O(3) and Grassmannian sigma models with the string solitons of this theory conjectured to be the constituent quarks of the confining phase.

The Niemi proposal has been pursued but only partially realized, mostly in that no explicit SCS-decomposed fermion analogous to the Luttinger spinon or chargon has been constructed in the 4D setting~\cite{ShabanovExtended,DittmannHeinzlWipf2001}. Without such a construction, the claim that SCS in gauge theory is structurally identical to SCS in Luttinger liquids remains at the level of an analogy. However, this is plausibly specific to the four-dimensional non-abelian gauge-theory setting. It is natural to ask whether the Faddeev--Niemi vision admits a concrete realization in other settings.

Our present work is motivated by results from a recent paper~\cite{Haddad2024} which shows that SCS follows from relativistic Dirac fermions in $(1+1)$ dimensions coupled to a chiral channel and a dynamical difermion field. Spinors are shown to decompose through a dressing as $\psi = \chi\Phi$, with each factor transforming under a distinct linear representation of the relevant Lorentz subgroup. The full in-medium symmetry group takes a product form through an Iwasawa decomposition, with the Lorentz group $\mathrm{SL}(2,\mathbb{R})$ of the elementary fermion becoming enhanced to a product $\mathrm{SL}(2,\mathbb{R})_1 \times \mathrm{SL}(2,\mathbb{R})_2$, acting separately on the two spinor components. This is the analog of the Faddeev--Niemi projective Lorentz action.

Thus, in some respects, the construction of~\cite{Haddad2024} realizes the Faddeev--Niemi idea in a setting that is restricted to the strong-pairing regime. The question we address in the present paper is whether the construction extends from this regime to a more comprehensive framework covering the full phase diagram of the Fierz-complete model. We also ask whether the three apparently distinct dressing constructions (Luttinger, Wilson-line, flat-connection) become unified in such an extension. Indeed, within our framework, we find that the three constructions arise as regime-specific projections of a single dressing matrix $U(x)$. The regime boundaries are characterized by which projection becomes singular.

The paper is organized as follows. Section~\ref{sec:self-consistent-spinors} sets up the self-consistent dressing and proves the trivialization theorem (Proposition~\ref{prop:trivialization}). Section~\ref{sec:local-fierz-complete} specializes to the Fierz-complete potential, develops the Magnus expansion for nearly-flat connections, and computes the non-abelian corrections to the emergent gauge field. Section~\ref{sec:interpretations} develops the three representations of the dressing and the Wilson-loop deconfinement diagnostic. Section~\ref{sec:synthesis} discusses the Niemi correspondence and its limits, and section~\ref{sec:conclusion} concludes. Two appendices, \ref{DiracConventions}, \ref{meanfield1}, present Dirac conventions and standard mean-field analysis.

\section{Fermion bilinears, quartic interactions, and Fierz identities}
\label{Fierz}

Generally, contact interactions in Dirac systems come in a variety of possible forms distinguished by their transformation under the Poincar\'e group. In this section, we will use the real representation for the gamma matrices
\begin{eqnarray}
\gamma^0 = \sigma_1, \;\;  \gamma^1 = -i \sigma_2, \;\; \gamma^5 =  \gamma^0 \gamma^1 = \sigma_3.
\end{eqnarray}
Quartic interactions are constructed from fermion bilinears which are local composite operators of the form
\begin{equation}
    \mathcal{O}(x) = \bar{\psi}(x)\, \Gamma\, \psi(x),
\end{equation}
where $\Gamma$ is a matrix in the Clifford algebra, such as $\mathbb{I}, \, \gamma^\mu,\,  \gamma^5,  \, \gamma^\mu \gamma^5, \, C$, with each bilinear product transforming under an irreducible representation of the Lorentz (or Poincar\'e) group. The key point is that this allows us to classify such products based on spin and parity. Standard bilinears along with their properties are listed in Table~\ref{bilinears}.
\begin{table}[h!]
\centering
\renewcommand{\arraystretch}{1.2}
\begin{tabularx}{\textwidth}{@{}lll >{\RaggedRight\arraybackslash}X @{}}
\toprule
Bilinear & Operator & Lorentz Transformation & Physical Interpretation \\ \midrule
Scalar & $\bar{\psi} \psi$ & Lorentz scalar & Parity-even condensate \\
Pseudoscalar & $\bar{\psi} \gamma^5 \psi$ & Lorentz scalar & Parity-odd condensate \\
Vector & $\bar{\psi} \gamma^\mu \psi$ & Lorentz vector & Conserved current \\
Axial vector & $\bar{\psi} \gamma^\mu \gamma^5 \psi$ & Axial vector & Chiral current \\
Difermion (scalar) & \( \psi^T C \psi \) & Scalar & Superconducting channel \\ \bottomrule
\end{tabularx}
\caption{Transformation properties and interpretation of common fermion bilinears under the Lorentz group in $(1+1)$d.}
\label{bilinears}
\end{table}
Such Lorentz-invariant operators have explicit form
\begin{equation}
\mathcal{L}_\textrm{int} = g \left( \bar{\psi} \Gamma \psi \right) \left( \bar{\psi} \Gamma \psi \right) \, ,
\end{equation}
with suitable contractions between spacetime indices to ensure that the operator transforms as a Lorentz scalar. The key distinction between fermion bilinears and their associated four-fermion interactions lies in their respective transformation properties and physical roles. Fermion bilinears transform non-trivially under the Lorentz group and serve as local order parameters, symmetry currents, or dynamical probes of fermionic degrees of freedom. In contrast, four-fermion interactions are composite scalar operators built from bilinear products and describe effective interactions among fermion pairs. The latter are required to be Lorentz scalars and thus involve contracted bilinears in appropriate representations. The classification of four-fermion interactions is thus determined by the quantum numbers of the bilinears they contain. 

Different four-fermion channels are related via Fierz transformations, which allow a given interaction to be re-expressed in different bilinear bases. This plays a critical role in identifying dominant condensation channels in models such as the Gross--Neveu or Nambu--Jona-Lasinio models. The practical point is that the reduced spinor and Clifford algebra structure in two dimensions simplifies the possible operator basis and results in strong constraints from Fierz identities. Explicitly, the Fierz rearrangement identity for bilinear contractions is
\begin{eqnarray}
(\bar{\psi}_i \Gamma^A \psi_j )(\bar{\psi}_k \Gamma_A \psi_l)
= -\frac{1}{2} (\bar{\psi}_i \Gamma^A \psi_l)(\bar{\psi}_k \Gamma_A \psi_j) \, ,
\end{eqnarray}
where $\Gamma^A$ runs over the Clifford basis: $\mathbb{I}, \,\gamma^\mu, \,\gamma^5$, and the field subscripts denote distinct field labels. Applying this to four-fermion interactions of a single flavor field $\psi$, we obtain
\begin{eqnarray}
(\bar{\psi} \gamma^\mu \psi)(\bar{\psi} \gamma_\mu \psi)
= - (\bar{\psi} \psi)^2 + (\bar{\psi}  \gamma^5 \psi)^2, \label{Fierz1} \\
(\bar{\psi} \gamma^\mu \gamma^5 \psi)(\bar{\psi} \gamma_\mu \gamma^5 \psi)
= - (\bar{\psi} \psi)^2 + (\bar{\psi}  \gamma^5 \psi)^2. \label{Fierz2}
\end{eqnarray}
Thus, the vector and axial-vector interaction channels are not independent in $(1+1)$ dimensions. In fact, all four-fermion interactions can be reduced to a basis consisting of scalar and difermion channels
\begin{eqnarray}
\mathcal{L}_{\text{int}} &=& g_s (\bar{\psi} \psi)^2 - g_d (\bar{\psi}  \gamma^5 \psi)^2 \\
&=& g_s (\bar{\psi} \psi)^2 - g_d (\bar{\psi}^C \psi)(\bar{\psi} \psi^C) \, ,
\end{eqnarray}
where in the second line we have related the pseudoscalar channel to the difermion operator via charge conjugation. The Fierz relations may be worked out directly in terms of the chiral components $\psi_{R,L} = \frac{1}{2}(1 \pm \gamma^5)\psi$. For instance, expanding expressions such as 
\begin{eqnarray}
(\bar{\psi} \psi)^2 =   ( \psi_R^\dagger \psi_L + \psi_L^\dagger \psi_R )^2 \, , \;\; \;\; 
(\bar{\psi} i \gamma^5 \psi)^2 =   ( \psi_R^\dagger \psi_L - \psi_L^\dagger \psi_R )^2 \, , 
\end{eqnarray}
while properly accounting for Grassmann anti-symmetrization, leads to eq.~\eqref{Fierz1}-\eqref{Fierz2}. Hence, such interactions naturally encode pairing structures between the left and right chiral modes, relevant to bosonization and renormalization group analyses in $(1+1)$-dimensional systems. Table~\ref{comparison} summarizes the main qualitative points regarding bilinears and their associated interactions. Given that the scalar interaction controls chiral-symmetry breaking and the difermion channel controls BCS-like pairing, one would expect that a combination of the two channels could form the basis for a simplified model of QCD in $(1+1)$d, at least in terms of its main qualitative features of its ground states and finite density phases transitions.
\begin{table}[h]
\centering
\renewcommand{\arraystretch}{1.3}
\begin{tabularx}{\textwidth}{@{}l X X@{}}
\toprule
Aspect & Fermion Bilinear &   Four-Fermion Interaction \\
\midrule
Clifford structure & $\mathbb{I}, \, \gamma^\mu, \, \gamma^5$ & Products of bilinears \\
Transformation & Irreducible Lorentz rep & Scalar composite (contractions) \\
Role & Source, order parameter & Interaction vertex \\
Symmetry breaking & Explicit or spontaneous & Drives condensate formation \\
Examples & 
\( \bar{\psi} \psi \), \( \bar{\psi} \gamma^\mu \psi \), \( \psi^T C \psi \) & 
\( (\bar{\psi} \psi)^2 \), \( (\bar{\psi} \gamma^\mu \psi)^2 \), \( | \psi^T C \psi |^2 \) \\
Redundancy & Linearly independent & Related via Fierz identities \\
\bottomrule
\end{tabularx}
\caption{Comparison between fermion bilinears and four-fermion interactions in $(1+1)$ dimensions.}
\label{comparison}
\end{table}

\section{Dressed spinors}
\label{sec:self-consistent-spinors}

In the present work we are interested in reframing Fierz complete models of Dirac fermions in $(1+1)$ dimensions. Such theories may be generically identified with the Lagrangian density for $N$ species of interacting fermions (for Dirac fermion conventions and a review of four-fermion interactions, see \autoref{DiracConventions} and~\autoref{Fierz})
\begin{eqnarray}
\hspace{-1pc} \mathcal{L}  &=& \sum_{a}  \bar{\psi}_a   \left( i \gamma^\mu  \partial_\mu + \mu \gamma^0  \right)   \psi_a  + \frac{g_s}{2}  \left(  \sum_a  \bar{\psi}_a     \psi_a  \right)^2  + \frac{g_d}{2}  \sum_{a, b} \left(   \bar{\psi}_a  C   \bar{\psi}^T_b  \right)  \left(   \psi^{T}_b   C     \psi_a   \right) ,  \label{Lag}
\end{eqnarray}
where the physical parameters $\mu$, $g_s$ and $g_d$ are the chemical potential, scalar and difermion couplings, respectively, and $C$ is the charge conjugation operator. Summation over flavor indices $a, \, b$, is shown where we adopt an off-diagonal pairing convention for the difermion term so that we may safely choose $C$ to be a symmetric matrix. Classical results are accurate when $N \gg 1$ in which case fermions may acquire a dynamically generated mass and difermion expectation value through the interactions with associated density, chiral, and difermion condensates (for background, see Refs.~\cite{Dirac1955,LavelleMcMullan1997}, and \autoref{meanfield1} for a review of mean-field analysis)
\begin{eqnarray}
\rho = \langle \bar{\psi } \gamma^0 \psi \rangle \, , \;\;\; \sigma =  \langle \bar{\psi} \psi \rangle \, , \;\;\;  \Delta = \langle  \psi^T C \psi \rangle   \, , \;\;\;  \bar{\Delta}   = \langle  \bar{\psi} C \bar{\psi}^T \rangle  \, . \label{potentials}
\end{eqnarray}
Our aim is to characterize the various types of reformulations of eq.~\eqref{Lag} that involve field redefinitions of the type $\psi = U \chi \equiv e^\Omega \chi$ with $\Omega \in \text{End}(\mathcal H_{\text{spinor}\otimes\text{int}})$; $\mathcal H_{\text{spinor}}\equiv$ spinor (Lorentz) rep space, e.g. $\mathbb C^{2}$ in (1+1)d, $\mathcal H_{\text{int}}\equiv$ internal space, e.g. $\mathbb C^{N_c} \otimes \mathbb C^{N_f}$, with $\psi(x)\in\mathcal H_{\text{spinor}}\otimes\mathcal H_{\text{int}}$. We identify $\chi$ as a dressed spinor that relates back to the fundamental fermions through 
\begin{eqnarray}
\chi =  (e^\Omega)^{-1}\,\psi = e^{-\Omega }\,\psi \, , 
\end{eqnarray}
in the case of a global $\Omega$, and 
\begin{eqnarray}
\chi(x) =  \mathcal P\exp\!\Big(-\!\int_{\gamma:x_0\to x}\! d\xi^\mu\,\mathcal A_\mu(\xi)\Big)\,\psi(x_0)\, ,
\end{eqnarray}
 otherwise, where $\mathcal{A}_\mu$ is the local generator determined by $\Omega$, and $\mathcal{P}$ is the path-ordering operator. The appeal of such transformations lies in their capacity to expose non-perturbative features from the standpoint of the original (undressed) degrees of freedom.

It is always algebraically possible to rewrite the fermion as a dressed field, \(\psi = e^\Omega \chi\), provided \(e^\Omega \) is invertible. However, this change of variables is physically useful only when \(\Omega \) is constructed from structures already present in the problem; here, from the potentials in eq.~\eqref{potentials}. In this case the Dirac operator takes the covariant form $\bar\chi\,i\gamma^\mu(\partial_\mu - iA_\mu^{\rm dress})\chi$ (with possibly a conformal factor) with $A_\mu^{\rm dress} =  i   (\partial_\mu U)U^{-1}$, which defines either a flat (trivializable) sector, with associated curvature $F_{+ -}^{\,\rm dress} = 0$, or produces an 
interaction tower for $F_{+ - }^{\, \rm dress} \neq 0$. By contrast, in the strictly free theory (no four-fermion couplings) there is no preferred choice of $U(x)$; any nontrivial $U$ merely manufactures a pure-gauge $A_\mu^{\rm dress}$ that can be removed again, in which case we may take $U=\mathbb{I}$ and $\chi=\psi$ without loss of generality.

In our work we must also draw an important distinction between \emph{external} versus \emph{self-consistent} enfolding of potentials into the spinor solutions. Throughout, ``externally-driven'' (non-autonomous) means that the amplitudes $\Delta$ and $\sigma$ are prescribed backgrounds for spinor solutions which must be determined independently, i.e., static or dynamic backgrounds whose profiles must be determined through separate equations of motion that ignore backreaction from the spinor field. In contrast, ``self-consistent'' (backreacted) means that $\Delta$ and $\sigma$ are determined together with the spinor solution, via stationarity (gap) conditions or dynamical equations. Concretely, if one introduces auxiliary bosons for the fermion bilinears via a Hubbard--Stratonovich step and writes the action schematically as
\[
S[\psi,\bar\psi, \Delta ,\sigma]
=\int d^2x\;\Big\{\bar\psi\,\mathcal{D}[ \Delta ,\sigma]\,\psi
- V_{\text{B}}( \Delta ,\sigma)\Big\},
\]
where $\mathcal{D}$ is the Dirac/BdG operator in the bosonic background and $V_{\text{B}}$ contains the quadratic and kinetic terms for $\Delta$ and $\sigma$, one finds the one-particle irreducible (1PI) effective action after integrating out the fermions 
\begin{equation}
\Gamma[ \Delta ,\sigma] = S_\mathrm{B}   [ \Delta ,\sigma] -    i      \ln\det  \! \left(  D[ \Delta  ,\sigma] \right) \,,
\end{equation}
where $S_\mathrm{B}$ encapsulates the local potentials/kinetic terms of $\Delta$ and $\sigma$. A configuration is \emph{self-consistent} if it obeys
\begin{equation}
\label{eq:self-consistency-stationarity}
\frac{\delta \Gamma}{\delta |\Delta| }=0,\qquad
\frac{\delta \Gamma}{\delta \beta }=0,\qquad
\frac{\delta \Gamma}{\delta \sigma }=0,  
\end{equation}
where $\Delta = |\Delta| e^{i \beta}$, together with the spinor equations in the same background. Equivalently, one could state that 
\begin{equation}
\label{eq:gap-meanfield}
\frac{|\Delta| \,e^{i\beta}}{g_d}= \langle \psi^T C\,\psi \rangle_{\Delta, \sigma} \, ,
\qquad
\frac{\sigma  }{g_s}=  \langle \bar\psi\psi  \rangle_{\Delta, \sigma} \, , 
\end{equation}
where the expectation values are computed in the background itself. Solving eq.~\eqref{eq:self-consistency-stationarity} (or eq.~\eqref{eq:gap-meanfield})
together with the reduced spinor equations yields the complete (classical) \emph{backreacted} solution. In a fully quantum treatment, self-consistency corresponds to integrating over the enlarged set of variables $(\psi,\bar\psi, \Delta, \Delta^* ,\sigma)$ in the path integral with the mean field treatment equivalent to the saddle point eq.~\eqref{eq:self-consistency-stationarity}. If the bosons are truly dynamical, e.g.\ $(\Box+m_\Delta^2) \Delta =\mathcal{J}_\Delta[\psi]$, $(\Box+m_\sigma^2)\sigma=\mathcal{J}_\sigma[\psi]$, with elementary spinors providing ``external'' source terms, then eq.~\eqref{eq:self-consistency-stationarity} will include such Euler--Lagrange equations with sources built from the same spinor solution; again, a backreacted (closed) system.

\subsection{Field redefinition for scalar potentials}

There is another sense in which we use \emph{self-consistency} in the present work. This is in regard to the local field redefinitions (discussed above) that enfold the difermion and scalar fields into the spinor itself. We will illustrate this procedure here using a simple scalar potential that minimally modifies the bare fermion mass $m$ without any additional spin couplings. We start from a generic $(1+1)$-dimensional Dirac theory 
\begin{equation}
\label{eq:origL}
\mathcal L_{\psi} \;=\; \bar\psi\,i\slashed{\partial}\psi \;-\; M_0(x)\,\bar\psi\psi \;+\; \mathcal L_{\rm int}[\psi]\,,
\end{equation}
where \(M_0(x)\) may include a bare mass and prescribed scalar potential and \(\mathcal L_{\rm int}\) contains any additional interactions (e.g. four-fermion terms). We introduce a real scalar field \(\sigma(x) := 2 \Omega(x)\) by the local Weyl-dressed decomposition
\begin{equation}
\label{eq:weylmap}
\psi \;=\; e^{-\sigma/2}\chi,\qquad \bar\psi \;=\; \bar\chi\,e^{-\sigma/2}\,.
\end{equation}
A direct substitution yields
\begin{equation}
\label{eq:Lprime}
\mathcal L'[\chi,\sigma] \;=\;
e^{-\sigma}\,\bar\chi\!\left(i\slashed{\partial}   -   \frac{i}{2}\slashed{\partial}\sigma\right)\!\chi
\;-\; e^{-\sigma}\,M_0(x)\,\bar\chi\chi \;+\; \mathcal L'_{\rm int}[\chi,\sigma] \,,
\end{equation}
with every operator in \(\mathcal L'_{\rm int}\) dressed by its Weyl weight; e.g.
\(g(\bar\psi\psi)^2 \mapsto g\,e^{-2\sigma}(\bar\chi\chi)^2\).
Equation~\eqref{eq:Lprime} is equivalently interpreted as the original theory
formulated on the conformally rescaled metric \(g_{\mu\nu}=e^{-2\sigma}\eta_{\mu\nu}\),
with the gradient term \(-\tfrac{i}{2}e^{-\sigma}(\partial_\mu\sigma)\bar\chi\gamma^\mu\chi\) 
playing the role of the 2D spin-connection or Weyl coupling. The status of \(\sigma\) admits several consistent specifications, none of which is needed in what follows. It may be taken as a Weyl compensator carrying a gauge redundancy \((\chi, \sigma) \sim (e^{-\alpha/2}\chi,\, \sigma - \alpha)\), in which case no new degree of freedom is introduced and variation with respect to \(\sigma\) yields a Ward identity rather than an independent equation of motion. It may be treated as a fixed spurion background, convenient for deriving local Weyl Ward identities and for studying the response to a prescribed conformal factor. It may be promoted to a dynamical dilaton with its own kinetic and potential terms, generated either by the loop-induced Polyakov term or by a Hubbard--Stratonovich decoupling of the scalar four-fermion channel. Or \(\sigma\) may simply weight the fermion operators by their Weyl dimensions as an external background not identified with any piece of \(M_0(x)\).

Varying \(\bar\chi\) in eq.~\eqref{eq:Lprime} gives the equation
\begin{equation}
\label{eq:chieom}
\Big(i\slashed{\partial}   -  \frac{i}{2}\slashed{\partial}\sigma - M_0(x)   \Big)\chi \;=\; 0 \, ,
\end{equation}
with the overall \(e^{-\sigma}\) factor dropping out of the equation of motion. In the special case where $M_0$ is purely imaginary, and $\slashed{\partial} \sigma = 2 i M_0$, one recovers the free (undressed) spinor, at least at the classical level: 
\begin{equation}
i \slashed{\partial} \chi = 0 \, . 
\end{equation}
The two--point functions relate multiplicatively:
\begin{equation}
\label{eq:greens}
G_\psi(x,y) \;\equiv\; \langle T\,\psi(x)\bar\psi(y)\rangle
\;=\; e^{-\sigma(x)/2}\,G_\chi(x,y)\,e^{-\sigma(y)/2}\,,
\end{equation}
so \(\chi \) is \emph{Weyl-dressed} by the local scale while \(\psi \) propagates in the canonical frame with
the derivative current coupling. In the compensator case the S--matrix is frame--independent.

\subsection{Field redefinition for general matrix potentials}
\label{sec:local-triv}

Here, we analyze the conditions that allow a general matrix-valued potential to be enfolded into the spinor field. This requires more care than the scalar case, since the potential can introduce non-trivial curvature. Consider the Dirac Lagrangian with a general matrix-valued background \(M(x)\in \mathrm{GL}(2,\mathbb{C})\),
\begin{equation}
\label{eq:L-psi-M}
\mathcal L[\psi;M] \;=\; \bar\psi\,\big(i\!\not\!\partial - M(x)\big)\,\psi \, , 
\end{equation}
and define the light-cone basis through \(x^\pm   =  t  \pm x   \), \(\gamma^\pm  =   \gamma^0\pm  \gamma^1\),
\(\gamma^5   =    \gamma^0   \gamma^1\), and chiral projectors \(P_{R, L}= \tfrac12(1\pm\gamma^5)\). We propose the following:

\begin{prop}[]
\label{prop:trivialization}
Define the matrix fields
\begin{equation}
\label{eq:Apm-def}
A_+(x) \,  := \,  \frac{1}{2\, i}\,P_R\,\gamma^-\,M(x),\qquad
A_-(x) \,   := \, \frac{1}{2\, i}\,P_L\,\gamma^+\,M(x) \, .
\end{equation}
Now assume that the \emph{zero-curvature}, or flatness, condition holds on a simply
connected domain \(\mathcal U\subset\mathbb{R}^{1,1}\), such that the associated field strength satisfies
\begin{equation}
\label{eq:flatness}
F_{+-} \;\equiv\; \partial_+ A_-(x) -  \partial_- A_+(x)  + [A_+(x),A_-(x)] \, = \,  0 \,  ,\quad \text{on  } \, \mathcal U.
\end{equation}
Then there exists \(U:  \mathcal U\to \mathrm{GL}(2,\mathbb{C})\) that satisfies
\begin{equation}
\label{eq:U-transport}
\partial_+ U \;=\; A_+ U,\qquad \partial_- U \;=\; A_- U,
\end{equation}
and consequently
\begin{equation}
\label{eq:Dirac-conj}
i\!\not\!\partial U \;=\; M\,U.
\end{equation}
With the local field redefinition \(\psi=U \chi\), \(\bar\psi=\bar\chi\,U^{-1}\),
the Lagrangian eq.~\eqref{eq:L-psi-M} becomes \(\mathcal L'=\bar\chi\,i\!\not\!\partial\,\chi\),
i.e. the Dirac operator is \emph{conjugated} to the free one.
\end{prop}

\begin{proof}
Multiply \(i\gamma^\mu\partial_\mu U = M U\) successively on the left by
\(P_R\gamma^-\) and \(P_L\gamma^+\).
Using \(\gamma^-\gamma^-=0\), \(\gamma^+\gamma^+=0\), and
\(P_R\,\gamma^- \gamma^+ = 2\,P_R\), \(P_L\,\gamma^+\gamma^- = -\,2\,P_L\)
(with our metric convention), we obtain the system eq.~\eqref{eq:U-transport}
with the definitions in eq.~\eqref{eq:Apm-def}.
The Frobenius integrability condition for eq.~\eqref{eq:U-transport} is precisely
eq.~\eqref{eq:flatness}. Moreover, since \(\mathcal U\) is simply connected, there exists
a smooth solution \(U\) unique up to a constant factor.
Conversely, any \(U\) solving eq.~\eqref{eq:U-transport} satisfies
\(i\!\not\!\partial U = M U\) by construction.
Finally,
\[
\bar\psi\,(i\!\not\!\partial - M)\psi
= \bar\chi\,U^{-1}\big(i\!\not\!\partial - M\big)U\,\chi
= \bar\chi\big(i\!\not\!\partial + U^{-1}i\!\not\!\partial U - U^{-1}MU\big)\chi
= \bar\chi\,i\!\not\!\partial\,\chi,
\]
where the last equality uses \eqref{eq:Dirac-conj}.
\end{proof}

\begin{remark}[]  
\label{rem:POexp}
When eq.~\eqref{eq:flatness} holds, one may write \(U\) as a path-ordered exponential
\(
U(x)=\mathcal P\exp\big(\int_{x_0}^{x} A_+\,dx'^+ + A_-\,dx'^-\big)\,U_0
\)
that is \emph{path-independent}. If \(F_{+-}\neq 0\), the same formula defines a
\emph{path-dependent} \(U\), where a Magnus expansion now applies. The significant difference in this case is that the resulting field redefinition is nonlocal and the transformed theory is free plus interactions built from the curvature \(F_{+-}\).
\end{remark}

\begin{lemma}[] 
\label{lem:EOM-det}
Under the local similarity \(\psi=U\chi\), \(\bar\psi=\bar\chi U^{-1}\),
the equation of motion \((i\slashed \partial - M)\psi=0\) is equivalent to \(i \slashed \partial \chi = 0\). At the functional level, \(\det(i \slashed \partial - M)= \det(i \slashed \partial)\) when \eqref{eq:flatness} holds, up to local counterterms from the regularization of \(\Tr\ln U\) and \(\Tr\ln U^{-1}\). 
\end{lemma}

\begin{cor}[Two solvable cases]
\label{cor:solvable-classes}\quad
\begin{itemize}

\item \emph{Pure vector/axial gradients.}  If \(M(x)= \tfrac{i}{2}\,\slashed{\partial}\sigma(x) + \tfrac{i}{2}\,\gamma^5 \slashed{\partial}\theta(x)\),
then \eqref{eq:flatness} holds identically and \(U(x)=\exp\big(-\tfrac12 \sigma(x) - \tfrac12 \theta(x)\,\gamma^5\big)\) solves \(i\!\not\!\partial U = M U\). This reproduces the familiar Weyl/chiral rescalings.

\item \emph{Chirally factorized backgrounds.} If \(M(x)=\gamma^- \mathcal M_+(x^+) + \gamma^+ \mathcal M_-(x^-)\) with \(\mathcal M_\pm\)
arbitrary \(2\times2\) matrices depending only on one light-cone coordinate,
then \(F_{+-}=0\) (since \(\partial_-A_+=\partial_+A_-=0\) and \([A_+,A_-]=0\)) and \(U(x)=\exp\big(\int^{x^+}\!A_+(s)\,ds\big)\,\exp\big(\int^{x^-}\!A_-(s)\,ds\big)\) trivializes the operator.

\end{itemize}
\end{cor}

\begin{remark}[]
\label{rem:curvature-EFT}
When \(F_{+-}\) is small, one may use the Magnus expansion
\(
U=\exp\{\Omega_1+\Omega_2+\cdots\}
\)
with
\(
\Omega_1=\int A_\mu dx^\mu,\quad
\Omega_2=\tfrac12\iint \Theta(s_1-s_2)\,[A(s_1),A(s_2)]\,ds_1ds_2,\ \ldots
\).
Conjugation by \(U\) maps the original theory to a free Dirac action
plus a (generally local, in gradient expansion) tower of vertices built from \(F_{+-}\)
and its derivatives. This provides a controlled effective field theory when the background is nearly trivializable.
\end{remark}

\section{Field redefinition for Fierz complete potentials}
\label{sec:local-fierz-complete}

We first summarize the scope of trivialization established in the previous section. The similarity transformation $\psi(x)=U(x)\,\chi(x)$ trivializes the Dirac operator
\(
(i\slashed{\partial}-M(x))\psi=0 \;\longrightarrow\; i\slashed{\partial}\,\chi=0
\)
\emph{iff} the matrix background $M(x)$ defines a \emph{flat} effective connection. In light-cone
notation (used for convenience), we defined the field components \(A_\pm(x)\) where the zero-curvature condition
\begin{align}
F_{+-} \equiv \partial_+A_- - \partial_-A_+ + [A_+,A_-] = 0
\end{align}
guarantees a local solution of $\partial_\pm U = A_\pm U$ so that $U^{-1}(i\slashed{\partial}-M)U=i\slashed{\partial}$.
This is the non-abelian Poincar\'e lemma: \emph{on a simply connected patch, flat connections are
locally pure gauge}~\cite{DemessieSaemann2014,HaydysIntroGauge}. The statement is coordinate-free; light-cone variables merely organize chiral projectors. The same light-cone system is formally the zero-curvature representation of $(1+1)$d integrable systems, in which the Dirac operator appears as the Zakharov-Shabat spectral problem and dressing transformations generate soliton solutions~\cite{ZakharovShabat1972,ZakharovShabat1974,AKNS1974}. The present use differs in aim and structure: no spectral parameter is introduced, and the connection is built from the condensate potentials with the goal of trivializing the operator rather than generating solutions. At the quantum level, local chiral/Weyl redefinitions of the fermions modify the path-integral measure and generate local counterterm anomalies via the Fujikawa Jacobian~\cite{AroucaCappelliHansson2024}.

A pure scalar coupling $M(x)=S(x)\,\mathbb{I}$ is not a gauge potential and cannot be removed by a similarity that only shifts a connection. Inserting $M=S\,\mathbb{I}$ into the definitions of $A_\pm$ above yields, in general, $F_{+-}\neq 0$ unless $S$ is constant; no $U(x)$ exists that maps $(i\slashed{\partial}-S)\psi=0$ to a free Dirac equation without introducing new derivative couplings. Physically, an inhomogeneous scalar (spatially varying mass) produces genuine effects (such as domain-wall zero modes, etc.), as exemplified by the Jackiw--Rebbi problem~\cite{JackiwRebbi1976,Charmchi2014}.
By contrast, backgrounds that are \emph{pure gradients} of vector/axial fields are classically pure gauge and can be removed locally, with the caveat that quantum anomalies reappear through the Fujikawa Jacobian~\cite{YaoFukusumi2019,AroucaCappelliHansson2024}. With $M=S\,\mathbb{I}$ and light-cone projectors $P_{R/L}$,
\(
A_+\propto P_R\,\gamma^-\,S,\quad A_-\propto P_L\,\gamma^+\,S.
\)
Because $S$ is a scalar function, $[A_+,A_-]$ does not cancel the cross-derivatives
$\partial_\pm S$ in general, so $F_{+-}\ne 0$ except when $S=\mathrm{constant}$.

To obtain the new dressed degrees of freedom in a natural way, we start by writing down the Dirac–HFB equations for the Nambu spinor $\Psi_N=(\psi,\psi^c)^T$ 
\begin{eqnarray}
\label{HFB}
\begin{pmatrix}
i\slashed{\partial} + \mu \gamma^0 - m - \sigma(x) & -\,\Delta(x) \\
\Delta^*(x) & i\slashed{\partial} - \mu \gamma^0 - m - \sigma(x)
\end{pmatrix}
\begin{pmatrix}\psi\\ \psi^c\end{pmatrix}=0 \, , 
\end{eqnarray}
where we consider a finite $\mu$ so that pairing of antiparticles is suppressed: $\bar{\Delta} \simeq 0$, $\Delta \ne 0$. To reduce this to a single equation, we incorporate charge conjugation into the difermion potential so that the latter is not a simple scalar potential but a $2 \times 2$ matrix potential that couples to spin. The condensed result is a Dirac-Hartree equation for a spinor coupled to scalar and difermion potentials:
\begin{eqnarray}
i \left[ \gamma^0\!  \left(\partial_t - i \mu \right) - \gamma^1\!  \left(\partial_x + i   \Delta(x) K \right) \right] \psi - \sigma(x) \,  \psi = 0 \, . \label{ModDirac}
\end{eqnarray}
The mass is now dynamically generated through $\sigma(x)$, and $K$ is the complex conjugation operator.

As an exercise, one might Fourier expand the solution in eq.~\eqref{ModDirac} in the chiral limit ($\sigma \to 0$) noting that the difermion potential naturally couples left and right Fourier modes. Taking the new fundamental fields to be noninteracting dressed modes that incorporate the background, neglecting backreaction for simplicity (passive background), allows us to express the Fourier amplitude for the spinor solution in eq.~\eqref{ModDirac} through the composite/dressed field $\chi = ( \chi_+ \,, \; \chi_- )^T$, where $\chi_\pm  \equiv     e^{  \mp \zeta  /2  }   \,  e^{  \mp i \beta /2  } \,  \phi^{\mp 1/2}  \,   \psi_\pm \,$ and $\psi = ( \psi_+ \,, \; \psi_-)^T$. Thus, one may express the dressed spinors in terms of the modified complex boost factors 
\begin{eqnarray}
 e^{\eta' /2} &=&  \sqrt{\frac{E + \mu + (p +  \Delta )}{E + \mu - (p +  \Delta )}}  \nonumber  \\
  &=&    \sqrt{\frac{E_+ }{E_-}}  \,     \sqrt{\frac{1   +   \mathrm{Re}( \Delta)/E_+ }{1 -    \mathrm{Re}( \Delta) / E_-}}  \,  \sqrt[4]{\frac{1   +     \mathrm{Im}( \Delta)^2/(E_+ +  \mathrm{Re}(  \Delta) )^2 }{1 +  \mathrm{Im} ( \Delta)^2 / (E_- -  \mathrm{ Re }( \Delta ))^2   }}  \, \nonumber \\
                   && \hspace{4pc} \times  e^{ i \left\{  \mathrm{tan}^{-1}\left[  \mathrm{Im} ( \Delta) /(E_+ +    \mathrm{ Re}(  \Delta )) \right]    -         \mathrm{tan}^{-1}\left[  \mathrm{Im}( \Delta) /(E_- -   \mathrm{Re} ( \Delta) ) \right]                              \right\} /2  }   \nonumber  \\
   &\equiv& e^{\eta/2} \, e^{\zeta/2}  \,  \phi^{1/2} \, e^{i \beta/2} \, . \label{dressed1} 
   \end{eqnarray}
 Factors in the last line are identified in order with those in the previous step, with kinematic variables $E_\pm = E + \mu \pm p \in \mathbb{R}$ and the difermion field $\Delta \in \mathbb{C}$. The first factor from the left (parametrized by $\eta$) in eq.~\eqref{dressed1} is the kinematic boost associated with the elementary spinor field. The second and third factors ($\zeta$ and $\phi$) encode a purely real amplitude generated by the background. The fourth factor ($\beta$) is a complex phase generated by the background. Thus, we see that this approach naturally incorporates the mean-field background at the classical level as a starting point for a complete quantum mechanical treatment (more background for this approach can be found in the appendices at the end of this paper and in our previous work~\cite{Haddad2024}).

 Turning now to the complete space-time problem with nontrivial coupling to the background, we write the equation of motion
\begin{equation}
\left( i  \, \slashed{\partial}  - M   \right) \psi = 0  \, , 
\end{equation}
where $M(x) =  - \gamma^0 \mu - \gamma^1 \Delta(x) K + \sigma(x)$. Following our analysis from the previous section, we have
\begin{eqnarray}
&&A_+  =   \frac{1}{2\, i}\,P_R\,\gamma^- M  =   i \left(  \begin{array}{ll}
\mu + \Delta K \;\;    &   - \sigma     \\
\;\; 0 \;\;    &  \;\; 0 \end{array}  \right) \, ,  \\
 &&A_-   = \frac{1}{2\, i}\,P_L\,\gamma^+\,M   =     i \left(  \begin{array}{ll}
\;\;  0  \;\;    &  \;\; \;\;\; \;\; \;\;   0     \\
 - \sigma  \;\;    &  \;\; \mu - \Delta K \end{array}  \right) \, ,
\end{eqnarray}
which give the field strength 
 \begin{eqnarray}
&&F_{+ -}  =    \left(  \begin{array}{ll}
 -i \partial_-( \Delta K + \mu) - \sigma^2  \;\;    &  \;\;  i \partial_-  \sigma + \sigma ( \mu - \Delta K)    \\
-i  \partial_+  \sigma - \sigma ( \mu + \Delta K) \;\;    &  \;\; \;\;  \;\;  i \partial_+( \Delta K - \mu)+ \sigma^2 \end{array}  \right) \, . 
\end{eqnarray}

Complete absorption (trivialization) of the potential into the dressed states requires \( F_{+-} = 0  \implies \)  
\begin{equation}
 \partial_- ( \Delta + \mu )=   i  \sigma^2 \; , \quad \partial_+ (\Delta - \mu )=  i  \sigma^2 \; , \quad    \partial_-  \sigma  = - i \sigma ( \mu - \Delta) \; , \quad    \partial_+  \sigma  =  i \sigma ( \mu + \Delta) \, . 
\end{equation}
 In the scalar region characterized by broken chiral symmetry (\( \sigma \ne 0  ;  \,        |\Delta|  , \, \mu \to 0  \)), with \( \Delta = |\Delta| e^{i \beta}\), the zero-curvature conditions lead to \( \beta \sim   b_+ {(x^+)}^n + b_- ({x^-)}^n \), where \(1/2 < n <1 \), \( |\Delta| \sim d_+ ({x^+})^{1-n} + d_- ({x^-)}^{1-n} \) and \( \sigma \sim \pm \sqrt{ d_\pm n /b_\pm } \), with the mass conditions \(\, m_\Delta \to \infty , \,   \partial_t^2 \sigma \sim - m_\sigma^2 \sigma   \). Lowest-order variations in the scalar field yield
 \begin{eqnarray}
 \sigma(x^+, x^-) \simeq \pm \sqrt{\frac{d_\pm n }{4 b_\pm}} \left[  1 - e^{- \frac{1}{2} (  \delta_{\mu  \nu}   \delta_{\beta}^\nu \delta_{\alpha}^\mu \,  d_\mu b_\nu \,  x^\alpha x^\beta )}  \right]\, , 
 \end{eqnarray}
 which are all consistent with expected properties of this regime. In contrast, for the difermion regime (\(\sigma \to 0 ;  \,  \Delta , \, \mu \ne 0 \)), the only way to satisfy the constraint equations is for $|\Delta|$ and $\beta$ constant. Departure from these conditions, i.e., $\Delta = \Delta(x^+, x^-)$, generates a nonzero curvature \( F_{+-} \sim \partial_\pm  \Delta \). In the crossover region, excitations of both fields are too massive to access at low energies, \(m_\sigma , \, m_\Delta \to \infty \), hence \( \sigma \simeq \Delta \simeq 0 \).

To go beyond the flat-connection limit when \(F_{+-}\neq 0\), we organize the dressing through the Magnus expansion. Writing the matrix potential as
\begin{equation}
M = \sigma\,\openone + \mathcal{M}_\Delta(\Delta),
\end{equation}
where \(\mathcal{M}_\Delta\) denotes the part sourced by the difermion channel, the realization of \(\mathcal{M}_\Delta\) depends on context: in a pseudoscalar mass representation, \(\mathcal{M}_\Delta = \Delta\,i\gamma^5\); in the Nambu--Gor'kov representation on the doubled spinor \(\Psi_N = (\psi, \psi^c)^T\),
\begin{equation}
\mathcal{M}_\Delta = \begin{pmatrix} 0 & \Delta \\ \Delta^* & 0 \end{pmatrix}\otimes \gamma^1,
\end{equation}
which is linear on the doubled space and avoids the antilinear bookkeeping otherwise required by the charge-conjugation operator \(K\). We adopt the doubled formulation throughout the remainder of this section, but suppress the Nambu indices when the structure is unambiguous.

Fixing a base point \(x_0\) and integrating along the rectangular path \(\gamma: x_0 \to (x^+, x_0^-) \to (x^+, x^-)\), we define
\begin{equation}
U(x) = \mathcal{P}_\gamma \exp\!\int_\gamma A_\mu\, dx^\mu = \exp\Omega,
\qquad
\Omega = \Omega_1 + \Omega_2 + \cdots,
\end{equation}
with
\begin{equation}
\Omega_1 = \int_{x_0^+}^{x^+}\! ds\, A_+(s,\, x_0^-) + \int_{x_0^-}^{x^-}\! dt\, A_-(x^+, t),
\end{equation}
\begin{equation}
\Omega_2 = -\tfrac{1}{2}\!\int_\gamma\!\!\int_\gamma\! ds_1\, ds_2\,\Theta(s_1 - s_2)\,[A(s_1), A(s_2)].
\end{equation}
For the rectangular contour, the bulk part of \(\Omega_2\) can be expressed as a surface integral of the curvature,
\begin{equation}
\Omega_2(x) = -\tfrac{1}{2}\!\int_{\Sigma(x_0 \to x)}\! d\sigma^{+-}\, F_{+-}(\xi) + (\text{edge commutators}),
\label{eq:Omega2-surface}
\end{equation}
so \(\Omega_2\) is first-order in \(F_{+-}\) and vanishes when the connection is flat.

The conjugated Dirac operator becomes
\begin{equation}
U^{-1}(i\slashed{\partial} - M)\,U = i\slashed{\partial} + \mathcal{V}_{\rm eff},
\qquad
\mathcal{V}_{\rm eff} \equiv i\,U^{-1}\slashed{\partial}\,U - U^{-1}\,M\,U.
\label{eq:Veff-def}
\end{equation}
Applying the Baker--Campbell--Hausdorff expansion to \(U = e^{\Omega_1 + \Omega_2 + \cdots}\) and collecting by order in \(F_{+-}\),
\begin{equation}
\mathcal{V}_{\rm eff} = \underbrace{\bigl(i\slashed{\partial}\Omega_1 - M\bigr)}_{\text{abelian / leading}}
\;+\; \underbrace{\bigl(i\slashed{\partial}\Omega_2 + \tfrac{i}{2}[\Omega_1,\,\slashed{\partial}\Omega_1]\bigr)}_{\mathcal{O}(F_{+-})}
\;+\; \mathcal{O}(F_{+-}^2,\, D_\pm F_{+-}).
\label{eq:Veff-orders}
\end{equation}
The Magnus expansion thus organizes the residual interaction as a local tower in \(F_{+-}\) and its covariant derivatives, providing a controlled effective field theory for nearly-flat backgrounds.

\subsection{Leading-order equations of motion and corrections}

At leading (abelian) order, the residual interaction vanishes when
\begin{equation}
i\slashed{\partial}\,\Omega_1 = M.
\label{eq:leading-Omega1}
\end{equation}
This is the analogue, at the Magnus level, of solving for the dressing in the chiral limit (eq.~\eqref{dressed1}); it determines \(\Omega_1\) as the abelian primitive of \(M\) and does not require flatness, but its consistency with the full non-abelian conjugation is controlled by the size of \(F_{+-}\) through the higher-order corrections in eq.~\eqref{eq:Veff-orders}.

The most general solution of eq.~\eqref{eq:leading-Omega1} requires \(\Omega_1\) to take values in the full algebra spanning the matrix space, which in our case is \(\mathfrak{gl}(2,\mathbb{C})\) (or its Nambu extension). A convenient basis is the Clifford basis \(\{\openone, \gamma^0, \gamma^1, \gamma^5\}\) with complex coefficients, which spans the \(2\times 2\) complex matrices over \(\mathbb{R}\):
\begin{equation}
\Omega_1 = (i a + e)\,\openone + (i b + f)\,\gamma^0 + (i c + g)\,\gamma^1 + (i d + h)\,\gamma^5,
\qquad a, b, c, d, e, f, g, h \in \mathbb{R}.
\label{eq:Omega1-basis}
\end{equation}
We have arranged the imaginary part of each coefficient to multiply the corresponding Clifford generator, anticipating that for real condensates \(M\) has real coefficients in the Clifford basis. Computing \(i\slashed{\partial}\Omega_1\) using the Clifford identities
\begin{equation}
(\gamma^0)^2 = \openone,\quad (\gamma^1)^2 = -\openone,\quad \gamma^0\gamma^1 = \gamma^5,\quad \gamma^0\gamma^5 = \gamma^1,\quad \gamma^1\gamma^5 = \gamma^0,
\end{equation}
and projecting onto each Clifford basis component, eq.~\eqref{eq:leading-Omega1} decouples into two independent sets of four real equations. With \(\partial_\pm \equiv \tfrac{1}{2}(\partial_0 \pm \partial_1)\) and \(\Delta = |\Delta|\,e^{i\beta} \equiv \Delta_R + i\Delta_I\), real-coefficient sector \((a, b, c, d)\), sourced by \((\sigma, \mu, \Delta_R)\):
\begin{equation}
\begin{aligned}
\partial_+(c - b) &= \tfrac{1}{2}\,\sigma, & \qquad
\partial_-(c + b) &= -\tfrac{1}{2}\,\sigma, \\
\partial_+(a + d) &= \tfrac{1}{2}(\mu + \Delta_R), &
\partial_-(a - d) &= \tfrac{1}{2}(\mu - \Delta_R).
\end{aligned}
\label{eq:EOM-real}
\end{equation}
The imaginary-coefficient sector \((e, f, g, h)\), sourced by \(\Delta_I\), gives:
\begin{equation}
\begin{aligned}
\partial_+(g - f) &= 0, & \qquad
\partial_-(g + f) &= 0, \\
\partial_+(e + h) &= \tfrac{1}{2}\,\Delta_I, &
\partial_-(e - h) &= -\tfrac{1}{2}\,\Delta_I.
\end{aligned}
\label{eq:EOM-imag}
\end{equation}
Two structural features are worth noting. First, the system decomposes into chiral combinations: each equation involves either \(\partial_+\) or \(\partial_-\) acting on one chirally-defined linear combination of the coefficients. This follows from the Clifford identity \(\gamma^0\gamma^5 = \gamma^1\), equivalently, the 2D relation \(\gamma^\mu\gamma^5 = \epsilon^{\mu\nu}\gamma_\nu\)). Second, the real-coefficient and imaginary-coefficient sectors decouple: the difermion amplitude \(|\Delta|\) enters through both \(\Delta_R\) and \(\Delta_I\), but the phase \(\beta = \arg\Delta\) sources only the \((e,f,g,h)\) sector. This separation is the algebraic precursor of the spin-charge factorization discussed in Sec.~\ref{sec:interpretations}.

Eq.~\eqref{eq:EOM-real} is over-determined as written since there are four equations for four functions, but two of them (the \(\sigma\)-sourced pair) place constraints on \(\sigma\) itself through the cross-derivative integrability conditions. Applying \(\partial_-\) to the first and \(\partial_+\) to the second of the \(\sigma\)-sourced equations yields
\begin{equation}
\partial_+\partial_-(c-b) = \tfrac{1}{2}\partial_-\sigma, \qquad
\partial_-\partial_+(c+b) = -\tfrac{1}{2}\partial_+\sigma,
\end{equation}
and adding gives \(\partial_+\partial_- c = \tfrac{1}{4}(\partial_-\sigma - \partial_+\sigma)\), \(\partial_+\partial_- b = -\tfrac{1}{4}(\partial_-\sigma + \partial_+\sigma)\), determining \(b\) and \(c\) up to harmonic functions of \(x^+\) and \(x^-\) separately. Analogous conditions hold in the \((a, d)\) and \((e, h)\) blocks. A background \((\sigma, \mu, \Delta)\) is \emph{trivializable at leading Magnus order} when these integrability conditions admit a global solution on the patch \(\mathcal{U}\).

The exact chiral decoupling above is a feature of the abelian limit. At \(\mathcal{O}(F_{+-})\), the second group of terms in eq.~\eqref{eq:Veff-orders} sources every Clifford component, including those that vanished at leading order. Projecting the homogeneous-source equations (the \(\partial_0 c - \partial_1 b = 0\) constraint in the real sector, and the analogous \(\partial_0 g - \partial_1 f = 0\) in the imaginary sector) onto the appropriate Clifford components yields
\begin{equation}
\partial_0 c - \partial_1 b = [F_{+-}]^{(5)} + \mathcal{O}(F_{+-}^2),
\label{eq:integrability-correction}
\end{equation}
where \([F_{+-}]^{(5)}\) denotes the \(\gamma^5\)-projection of the curvature matrix evaluated at \(x\). Adding this to the \(\sigma\)-sourced equation \(\partial_1 c - \partial_0 b = \sigma\) (the \(\openone\)-projection at leading order) gives the corrected chiral pair
\begin{equation}
\partial_+(c - b) = \tfrac{1}{2}\,\sigma + \tfrac{1}{2}[F_{+-}]^{(5)},
\qquad
\partial_-(c + b) = -\tfrac{1}{2}\,\sigma + \tfrac{1}{2}[F_{+-}]^{(5)}.
\end{equation}
The curvature thus enters as a \emph{common-mode shift} in both \(\partial_\pm\) equations, while the scalar condensate enters as a \emph{differential-mode shift}. Analogous corrections appear in the \((a,d)\) and \((e,h)\) blocks, with \([F_{+-}]^{(0)}\), \([F_{+-}]^{(1)}\) sourcing the \(\gamma^0\) and \(\gamma^1\) projections respectively.

Finally, we note that the \(F_{+-} \to 0\) limit recovers the chirally-factorized equations~\eqref{eq:EOM-real}--\eqref{eq:EOM-imag}. The Magnus expansion is therefore controlled by the dimensionless ratio of \(F_{+-}\) to the characteristic scale of the source \((\sigma, \mu, \Delta)\), and the leading correction has a simple physical meaning: in 2D, the abelian dressing rotates left- and right-moving sectors independently, and curvature couples them.

\subsection{Limiting regimes and the Clifford-support}

We now examine the two limiting regimes where one source dominates, and a crossover where both are small, using eqs.~\eqref{eq:EOM-real}--\eqref{eq:EOM-imag} directly.

\begin{enumerate}

\item \textbf{Scalar (chiral) regime: \(\sigma \neq 0,\ |\Delta|, \mu \to 0\).}
The \((a,d)\) and \((e,f,g,h)\) blocks admit trivial (zero) solutions. The non-trivial content is the \((b,c)\) pair sourced by \(\sigma\)
\begin{equation}
\partial_+(c - b) = \tfrac{1}{2}\,\sigma,
\qquad
\partial_-(c + b) = -\tfrac{1}{2}\,\sigma.
\end{equation}
For a chiral background of traveling-wave form \(\sigma(\xi) = \sigma(k x - \omega t)\), introducing \(\xi^\pm = (k \mp \omega)x^\pm / 2\) and integrating gives bounded real solutions
\begin{equation}
c(x) - b(x) = \tfrac{1}{2}\!\int^{x^+}\!\!\sigma\,dx'^+, \qquad
c(x) + b(x) = -\tfrac{1}{2}\!\int^{x^-}\!\!\sigma\,dx'^-,
\end{equation}
from which \(b\) and \(c\) are recovered. The dressing matrix in this regime is
\begin{equation}
\Omega_1 = i\bigl[b(x)\gamma^0 + c(x)\gamma^1\bigr],
\end{equation}
which encodes a position-dependent boost in the \((x^0, x^1)\)-plane where the conjugation \(\psi = e^{\Omega_1}\chi\) absorbs the chiral condensate as a local Lorentz rotation.

\item \textbf{Difermion regime: \(|\Delta| \neq 0,\ \mu \neq 0,\ \sigma \to 0\).}
The \((b,c)\) block now admits trivial solutions, and the active equations are
\begin{equation}
\partial_+(a + d) = \tfrac{1}{2}(\mu + \Delta_R), \qquad
\partial_-(a - d) = \tfrac{1}{2}(\mu - \Delta_R),
\end{equation}
\begin{equation}
\partial_+(e + h) = \tfrac{1}{2}\,\Delta_I, \qquad
\partial_-(e - h) = -\tfrac{1}{2}\,\Delta_I.
\end{equation}
Assuming traveling-wave backgrounds \(\mu(\xi)\) and \(\Delta(\xi)\) with the chiral decompositions \(\xi_r = k_r x - \omega_r t\) (for the real part) and \(\xi_i = k_i x - \omega_i t\) (for the imaginary part), bounded solutions take the form
\begin{equation}
a + d = \tfrac{1}{2}\!\int^{x^+}\!\!(\mu + \Delta_R)\,dx'^+, \qquad
a - d = \tfrac{1}{2}\!\int^{x^-}\!\!(\mu - \Delta_R)\,dx'^-,
\end{equation}
\begin{equation}
e + h = \tfrac{1}{2}\!\int^{x^+}\!\!\Delta_I\,dx'^+, \qquad
e - h = -\tfrac{1}{2}\!\int^{x^-}\!\!\Delta_I\,dx'^-.
\end{equation}
The dressing in this regime carries non-trivial \(\openone\) and \(\gamma^5\) contributions: the imaginary parts of these coefficients come from \((a, d)\), sourced by \(\mu\) and \(\Delta_R\), and the real parts come from \((e, h)\), sourced by \(\Delta_I\). The components \(f, g\) (real parts of \(\gamma^0, \gamma^1\)) and \(b, c\) (imaginary parts of \(\gamma^0, \gamma^1\)) all vanish. Physically, \(\Omega_1\) encodes a position-dependent \(U(1)\) phase (charge, generated by \(\openone\)) together with an axial \(U(1)_A\) phase (spin, generated by \(\gamma^5\)), in line with the spin-charge factorization developed below.

When \(|\Delta|\) and \(\beta\) are spacetime-constant, the chiral integrals reduce to linear functions of \(x^\pm\), giving the simple boost \(\times\) phase form anticipated by eq.~\eqref{dressed1}. Departures from constancy generate non-zero \(F_{+-} \sim \partial_\pm \Delta\) and require the corrections of eq.~\eqref{eq:integrability-correction}.

\item \textbf{Crossover region: \(\sigma, |\Delta|, \mu\) all small but non-zero.}
The two sectors couple through the leading non-abelian correction \([F_{+-}]^{(\Gamma)}\). Both chiral and difermion modes are heavy in this regime (\(m_\sigma, m_\Delta \gg \) IR scales) and the appropriate description is the truncated Magnus expansion with both \(\sigma\) and \(\Delta\) treated perturbatively. See Sec.~\ref{subsec:intermediate-regime} for the spectral and EFT analysis of this regime.

\end{enumerate}

All three regimes share the property \(F_{+-} \simeq 0\), but for two distinct physical reasons. In the scalar and difermion regimes, the dominant source is large and approximately uniform, so its gradients are small and the leading-order chirally-decoupled equations apply directly. In the crossover region, both sources are small to begin with, and the heavy in-medium masses \(m_\sigma, m_\Delta\) suppress fluctuations of the corresponding fields at low energies. In each case, corrections to \(F_{+-}\) come from small deviations from the uniform or vanishing-field limits and can be accounted for through the Magnus expansion of eq.~\eqref{eq:integrability-correction}.

Trivializability at leading Magnus order is not, by itself, a regime indicator: all three regimes are trivializable to this order. What distinguishes them is the Clifford-basis support of \(\Omega_1\), summarized in table~\ref{tab:clifford-support}. In the two ordered regimes the dressing inhabits a proper Lie subalgebra of the Clifford algebra, and the spin-charge factorization in the difermion regime is the statement that \(\Omega_1\) commutes with itself in the abelian subalgebra \(\mathrm{span}\{\openone, \gamma^5\}\). In the crossover region all eight coefficients are sourced perturbatively, the two blocks couple at \(\mathcal{O}(F_{+-})\) through \([F_{+-}]^{(\Gamma)}\) in eq.~\eqref{eq:integrability-correction}, and no subalgebra restriction applies; the spectral and effective-field-theoretic structure of this regime is the subject of section~\ref{subsec:intermediate-regime}.

\begin{table}[t]
\centering
\begin{tabularx}{\textwidth}{ll l X} 
\toprule
Regime & Active coefficients & Subalgebra & Content of \(\Omega_1\) \\ 
\midrule
Chiral & \(b, c\) & \(\mathrm{span}\{\gamma^0, \gamma^1\}\) (boost) & position-dependent boost absorbing \(\sigma\) \\ 
\addlinespace 
Difermion & \(a, d\) and \(e, h\) & \(\mathrm{span}\{\openone, \gamma^5\}\) (abelian) & \(U(1) \times U(1)_A\) charge and axial phases \\ 
\addlinespace
Intermediate & all eight & full Clifford basis & blocks coupled at \(\mathcal{O}(F_{+-})\) \\ 
\bottomrule
\end{tabularx}
\caption{Clifford-basis support of the dressing generator \(\Omega_1\) across the three regimes.}
\label{tab:clifford-support}
\end{table}

\subsection{Emergent gauge field and spin-charge confinement}

The Clifford-support diagnostic above can be reformulated in terms of an emergent gauge field that mediates the binding of spin and charge degrees of freedom in the dressed-spinor representation. In~\cite{Haddad2024}, this gauge field arises when the spin-charge decomposition \(\chi_\pm = e^{\mp\zeta/2} \phi_\pm\) is substituted into the fermion kinetic term and local fluctuations of the difermion phase \(\beta\) and the overall number phase \(\theta_N\) are introduced. Local gauge invariance is recovered by promoting ordinary derivatives to covariant ones, \(\partial_\mu \to \partial_\mu - i \mathcal{A}_\mu\), with the gauge components transforming as
\begin{equation}
\mathcal{A}_0 \to \mathcal{A}_0 - (\partial_t \theta_N + \partial_x \beta), \qquad
\mathcal{A}_1 \to \mathcal{A}_1 - (\partial_x \theta_N + \partial_t \beta).
\label{eq:gauge-field-defs}
\end{equation}
The field strength is then
\begin{equation}
\mathcal{E} \;\equiv\; -\mathcal{F}^{01} \;=\; (\partial_t^2 - \partial_x^2)\,\beta,
\label{eq:gauge-field-strength}
\end{equation}
generated by fluctuations of the difermion phase \(\beta\) and independent of the overall number phase \(\theta_N\). Physically, \(\mathcal{E}\) measures the binding between the spin and charge degrees of freedom: when \(\mathcal{E}\) is large, the spinon and chargon components of the dressed fermion are tightly bound; when \(\mathcal{E}\) vanishes, they propagate independently~\cite{Haddad2024}.

The behavior of \(\mathcal{E}\) across the three regimes is consistent with the Clifford-support diagnostic and gives a physical interpretation of it.

\begin{enumerate}

\item In the chiral phase, \(\beta\) is ill-defined (the corresponding \(U(1)\) is unbroken), and quantum fluctuations of \(\beta\) are large. The field strength \(\mathcal{E}\) is correspondingly large, and the spin and charge degrees of freedom are tightly bound. This is the confining phase for the spin-charge structure.

\item In the difermion phase, \(\beta\) is a stiff Goldstone mode of the broken \(U(1)\) symmetry. Density and phase fluctuations of \(\Delta\) decouple at low energies, and \(\mathcal{E} \to 0\) in the infrared. The spin and charge degrees of freedom propagate independently. This is the deconfining phase for the spin-charge structure, and is the regime in which the Lorentz doubling \(\mathrm{SL}(2,\mathbb{R}) \to \mathrm{SL}(2,\mathbb{R})_1 \times \mathrm{SL}(2,\mathbb{R})_2\) of~\cite{Haddad2024} is realized.

\item In the intermediate regime, both \(\sigma\) and \(\Delta\) are heavy. The fluctuations of \(\beta\) are partially suppressed by the difermion mass \(m_\Delta\), and \(\mathcal{E}\) takes intermediate values regulated by this mass. The spin and charge degrees of freedom are partially bound, and the regime interpolates between the confining and deconfining limits.

\end{enumerate}
The progression from the chiral phase to the difermion phase, via the intermediate regime, is therefore a deconfinement transition for the spin-charge degrees of freedom. The gauge field \(\mathcal{A}_\mu\) plays the role of a binding field that is strong in the chiral phase, weak in the difermion phase, and intermediate in the crossover. The Lorentz doubling and the \(\mathcal{PT}\)-symmetry breaking diagnostic of~\cite{Haddad2024} are the symmetry-group signatures of this deconfinement transition.

The emergent gauge field \(\mathcal{A}_\mu\) of eq.~\eqref{eq:gauge-field-defs} must be distinguished from the composite connection \(A_\mu^{\rm dress} = i(\partial_\mu U) U^{-1}\) built from the dressing matrix \(U\). The two are related but not identical. The composite connection \(A_\mu^{\rm dress}\) is defined for any dressing \(U\) and takes values in the full Clifford algebra; its Clifford-component support tracks the regime structure as discussed above. The emergent gauge field \(\mathcal{A}_\mu\) is the projection of \(A_\mu^{\rm dress}\) onto the abelian subalgebra \(\mathrm{span}\{\openone, \gamma^5\}\) and is the appropriate object when the dressed fields carry definite charge and axial-phase content. The two coincide in the difermion phase, where the Clifford support of \(\Omega_1\) is exactly this subalgebra and half-infinite Wilson lines in the abelian \(\gamma^5\) sector define gauge-invariant asymptotic single-particle states. In the chiral phase, \(A_\mu^{\rm dress}\) lives in the boost subalgebra \(\mathrm{span}\{\gamma^0, \gamma^1\}\), the natural Wilson-type object is the closed Wilson loop in that sector, and its area-law scaling signals the confinement of spin and charge as developed in section~\ref{subsec:wilson-lines}.

\subsection{Non-abelian corrections to the emergent gauge field}
The relation between \(\mathcal{A}_\mu\) and \(A_\mu^{\rm dress}\) just discussed has a quantitative consequence. Specifically, outside the difermion phase, the composite connection acquires Yang-Mills like corrections that promote it from an abelian gauge field to a non-abelian one. The calculation below makes the structure explicit and shows in which regime the corrections are largest.

The composite connection is computed from \(U = \exp(\Omega)\), with \(\Omega = \Omega_1 + \Omega_2 + \cdots\) the Magnus expansion of eq.~\eqref{eq:integrability-correction}. Using the standard identity \((\partial_\mu e^\Omega)\,e^{-\Omega} = \int_0^1 \!ds\,(\mathrm{Ad}\,e^{s\Omega})\,\partial_\mu \Omega\) and expanding the adjoint action in powers of \(\Omega\), one obtains the Baker--Campbell--Hausdorff series
\begin{equation}
A_\mu^{\rm dress} \;=\; i\,\partial_\mu \Omega \;+\; \frac{i}{2}\,[\Omega,\, \partial_\mu \Omega] \;+\; \frac{i}{6}\,\big[\Omega,\,[\Omega,\, \partial_\mu \Omega]\big] \;+\; \mathcal{O}(\Omega^4).
\label{eq:A-dress-BCH}
\end{equation}
Substituting \(\Omega = \Omega_1 + \Omega_2\) and keeping terms through \(\mathcal{O}(F_{+-}^{(0)})\) (i.e.\ linear in \(\Omega_2\)) gives
\begin{equation}
\begin{aligned}
A_\mu^{\rm dress} \;=\;& \underbrace{i\,\partial_\mu \Omega_1}_{\mathcal{O}(1),\;\text{abelian if }\Omega_1\in\text{abelian sub}}
\;+\; \underbrace{\frac{i}{2}\,[\Omega_1,\, \partial_\mu \Omega_1]}_{\mathcal{O}(\Omega_1^2),\;\text{non-abelian}} \\[4pt]
&\;+\; \underbrace{i\,\partial_\mu \Omega_2 \;+\; \frac{i}{2}\big([\Omega_1, \partial_\mu \Omega_2] + [\Omega_2, \partial_\mu \Omega_1]\big)}_{\mathcal{O}(F_{+-}^{(0)})}
\;+\; \cdots.
\end{aligned}
\label{eq:A-dress-expanded}
\end{equation}
The first term reproduces the leading abelian gauge field. The second is a non-abelian correction that is independent of the curvature \(F_{+-}^{(0)}\) and survives whenever \(\Omega_1\) lives in a non-abelian subalgebra of the Clifford algebra. The third group of terms gives the genuine curvature corrections. The structure constants of the Clifford algebra, computed in the basis of section~\ref{DiracConventions}, are
\begin{equation}
[\gamma^0, \gamma^1] = 2\gamma^5, \qquad
[\gamma^0, \gamma^5] = 2\gamma^1, \qquad
[\gamma^1, \gamma^5] = 2\gamma^0,
\label{eq:clifford-commutators}
\end{equation}
which form the Lie algebra \(\mathfrak{sl}(2,\mathbb{R})\) (with \(\openone\) commuting with everything and decoupling). The non-abelian sector of the composite connection therefore takes values in \(\mathfrak{sl}(2,\mathbb{R})\), consistent with the Iwasawa decomposition of the in-medium dressed-spinor symmetry group of~\cite{Haddad2024}.

The behavior of eq.~\eqref{eq:A-dress-expanded} in each regime follows from the Clifford-support pattern established in the diagnostic above, which we now comment on here.

\subparagraph{Difermion phase.} \(\Omega_1 \in \mathrm{span}\{\openone, \gamma^5\}\) is in an abelian subalgebra. The commutator \([\Omega_1, \partial_\mu \Omega_1]\) vanishes identically, and the curvature corrections involving \(\Omega_2\) also remain in the abelian subalgebra to leading order (since \(F_{+-}^{(0)}\) itself is abelian-valued when \(\Omega_1\) is). The composite connection is therefore strictly abelian in this regime:
\begin{equation}
A_\mu^{\rm dress}\big|_{\rm dif} \;=\; i\,\partial_\mu \Omega_1 \;+\; i\,\partial_\mu \Omega_2 \;+\; \mathcal{O}(F_{+-}^2),
\qquad \Omega_1, \Omega_2 \in \mathrm{span}\{\openone, \gamma^5\}.
\end{equation}
This is the regime in which \(\mathcal{A}_\mu\) of~\cite{Haddad2024} is the complete gauge field; no non-abelian elements are generated.

\subparagraph{Chiral phase.} \(\Omega_1 = i[b(x)\gamma^0 + c(x)\gamma^1]\) lives in the non-abelian boost subalgebra \(\mathrm{span}\{\gamma^0, \gamma^1\}\). Direct calculation using eq.~\eqref{eq:clifford-commutators} gives
\begin{equation}
\frac{i}{2}\,[\Omega_1, \partial_\mu \Omega_1] \;=\; i\,\big(c\,\partial_\mu b - b\,\partial_\mu c\big)\,\gamma^5,
\label{eq:chiral-noabelian-correction}
\end{equation}
so the composite connection acquires a \(\gamma^5\)-valued component sourced by the Wronskian of \((b, c)\). The corresponding field strength is
\begin{equation}
F^{\rm dress}_{tx}\Big|_{\rm chiral} \;=\; \partial_t A_x - \partial_x A_t - i[A_t, A_x] \;=\; 2i\,\big(\partial_x b\,\partial_t c - \partial_t b\,\partial_x c\big)\,\gamma^5 \;+\; \cdots,
\label{eq:F-dress-chiral}
\end{equation}
where the commutator term \(-i[A_t, A_x]\) is precisely the Yang--Mills self-coupling contribution. Eq.~\eqref{eq:F-dress-chiral} is a genuine non-abelian field strength: the Jacobian of \((b, c)\) sources a \(\gamma^5\)-valued curvature, even when the original \(F_{+-}^{(0)}\) is zero.

\subparagraph{Intermediate regime.} The dressing has support on the full Clifford basis \(\{\openone, \gamma^0, \gamma^1, \gamma^5\}\). Each of the three commutators in eq.~\eqref{eq:clifford-commutators} contributes, and the non-abelian correction \(\tfrac{i}{2}[\Omega_1, \partial_\mu \Omega_1]\) has non-zero components in all of \(\gamma^0\), \(\gamma^1\), and \(\gamma^5\):
\begin{equation}
\begin{aligned}
\frac{i}{2}\big[\Omega_1, \partial_\mu \Omega_1\big]\Big|_{\gamma^0}
&\;=\; i(d\,\partial_\mu c - c\,\partial_\mu d) + (g\,\partial_\mu h - h\,\partial_\mu g) + i\,\text{(mixed)}, \\
\frac{i}{2}\big[\Omega_1, \partial_\mu \Omega_1\big]\Big|_{\gamma^1}
&\;=\; i(d\,\partial_\mu b - b\,\partial_\mu d) + (f\,\partial_\mu h - h\,\partial_\mu f) + i\,\text{(mixed)}, \\
\frac{i}{2}\big[\Omega_1, \partial_\mu \Omega_1\big]\Big|_{\gamma^5}
&\;=\; i(c\,\partial_\mu b - b\,\partial_\mu c) + (f\,\partial_\mu g - g\,\partial_\mu f) + i\,\text{(mixed)},
\end{aligned}
\label{eq:intermediate-noabelian}
\end{equation}
where ``mixed'' refers to terms involving products of real and imaginary coefficients (e.g.\ \(b\,\partial_\mu g\), \(g\,\partial_\mu b\), and so on). The non-abelian structure is most pronounced in this regime, both because all three sectors are active and because the curvature corrections from \(\Omega_2\) further activate cross-block couplings absent at strict \(F_{+-}^{(0)} = 0\).

The relation of this non-abelian residue to the Niemi proposal for spin-charge separation as a mechanism for confinement is discussed in section~\ref{sec:synthesis}.

\section{Representations of the dressed spinor} 
\label{sec:interpretations}

As we have discussed so far, the dressed formulation $\psi(x) = U(x)\,\chi(x)$ developed in section~\ref{sec:local-fierz-complete} acquires structure specific to each of the fermion regimes. Each of these admits a natural representation, which we discuss now. These are: 1) the spin-charge factorization $U = U_c U_s$ of~\cite{Haddad2024}, most concrete in the difermion phase where it is exact at leading Magnus order; 2) the Wilson-line/Wilson-loop framework, which tracks the spin-charge deconfinement transition through the scaling of Wilson-type objects; and 3) the flat-connection holonomy, the structural language in which both regime-specific descriptions live.

\subsection{Spin-charge factorization}
\label{sec:spin-charge}

The SCS factorization $U = U_c U_s$ is the natural interpretation of the dressing in the difermion phase. In that phase $\Omega_1$ lives in the abelian Clifford subalgebra $\mathrm{span}\{\openone, \gamma^5\}$. Here, the bosonized spin-charge variables of~\cite{Haddad2024} provide the appropriate low-energy degrees of freedom. We extend the construction here and connect it to the emergent abelian gauge field that mediates the spin-charge binding outside the difermion phase.

\subsubsection{Relation to the standard bosonization derivation}

Both the present construction and the standard bosonization treatment of SCS in $(1+1)$d~\cite{LutherEmery1974,Haldane1981,Giamarchi2004} derive SCS as a long-wavelength effect around an appropriate background. The standard derivation linearizes the dispersion around the Fermi points $\pm k_F$ of the non-interacting Fermi sea and bosonizes the resulting left- and right-movers. SCS then appears as the decoupling of charge and spin bosonic sectors at the level of quadratic fluctuations around this background. The present construction takes the underlying theory to be a relativistic Dirac multiplet, so no linearization is needed. The background here is the mean-field condensate rather than the Fermi sea. The two approaches both essentially identify SCS in the long-wavelength fluctuations around a chosen background. The connection to the standard bosonization variables is made explicit through the exponential vertex-operator factors appearing in our chiral decomposition, which carry the fermion statistics of the dressed-spinor.

Thus, following~\cite{Haddad2024}, in the SCS regime we may separate the internal space as $\mathcal{H}_{\rm int} = \mathcal{H}_{U(1)_c} \otimes \mathcal{H}_{G_s}$, with generator $Q$ for $U(1)_c$ and generators $S^a$ for the spin group $G_s$. The dressing factorizes into commuting pieces,
\begin{equation}
U(x) = U_c(x)\,U_s(x), \qquad U_c = e^{-i\theta_N(x) Q}, \qquad U_s = e^{-\frac{1}{2}\beta(x)\gamma^5},
\label{eq:SCS-factorization}
\end{equation}
so that the charge phase $\theta_N$ generates a $U(1)$ rotation and the spin/boost rapidity $\beta$ generates an $\mathrm{SL}(2,\mathbb{R})$ Lorentz boost acting on $\chi$. In the chiral basis the dressed spinor reads
\begin{equation}
\psi(x) = e^{-i\theta_N(x)}\,e^{-\frac{1}{2}\beta(x)\gamma^5}\,\chi(x) = \begin{pmatrix} e^{-\frac{1}{2}(\beta+2i\theta_N)} & 0 \\ 0 & e^{+\frac{1}{2}(\beta-2i\theta_N)} \end{pmatrix} \chi(x).
\label{eq:SCS-chiral-basis}
\end{equation}
The composite connection $A_\mu^{\rm dress} = i(\partial_\mu U)U^{-1}$ splits along the factorization as
\begin{equation}
A_\mu^{\rm dress} = A_\mu^c\,Q + A_\mu^{s,a}\,S^a, \qquad A_\mu^c = \partial_\mu \theta_N, \qquad A_\mu^s = i(\partial_\mu U_s)\,U_s^{-1}\, . 
\label{eq:SCS-connection-split}
\end{equation}

\subsubsection{Charge and spin content of the potential}

The analysis above is purely structural in that it does not address the interactions directly. To do this, we start by decomposing the matrix potential as
\begin{equation}
M[\sigma, \Delta] = \underbrace{\sigma\,\openone + \mathcal{M}_\Delta^c}_{\text{commutes with }S^a} \;+\; \underbrace{\mathcal{M}_\Delta^s}_{\text{transforms in adjoint of }G_s},
\end{equation}
where $\mathcal{M}_\Delta^c$ couples to the $U(1)_c$ sector and $\mathcal{M}_\Delta^s$ carries non-abelian spin structure (off-diagonal Nambu--Gor'kov terms). When $[M, Q] = 0$ the charge and spin dressings commute, $[A_\mu^c Q, A_\nu^s] = 0$. With the light-cone definition $A_\pm = \tfrac{i}{2} P_{R/L} \gamma^\mp M$, the field strength decomposes as
\begin{equation}
F_{+-}^{\rm dress} = F_{+-}^c\,Q + F_{+-}^{s,a}\,S^a,
\end{equation}
with abelian and non-abelian components
\begin{equation}
F_{+-}^c = \partial_+ A_-^c - \partial_- A_+^c, \qquad F_{+-}^s = \partial_+ A_-^s - \partial_- A_+^s - i[A_+^s, A_-^s].
\end{equation}
The full flatness condition then factorizes:
\begin{equation}
F_{+-}^{\rm dress} = 0 \quad \Longleftrightarrow \quad F_{+-}^c = 0 \;\; \text{and} \;\; F_{+-}^s = 0.
\end{equation}
A pure scalar deformation $M = \sigma(x)\,\openone$ gives $F_{+-}^s = 0$ but generically $F_{+-}^c \neq 0$ unless $\sigma$ is constant, while a chirally factorized spin potential gives $F_{+-}^s = 0$ as in the leading-order analysis of section~\ref{sec:local-fierz-complete}, recovering exact trivialization.

Now consider fixing the rectangular path $\gamma$ of section~\ref{sec:local-fierz-complete} and write $U = e^\Omega$ with $\Omega = \Omega_c + \Omega_s$ and $\Omega_{(\cdot)} = \Omega_{1,(\cdot)} + \Omega_{2,(\cdot)} + \cdots$. If $[A_\mu^c Q, A_\nu^s] = 0$ then to $\mathcal{O}(\varepsilon^2)$ in the nearly-flat regime the mixed commutators drop out, and
\begin{equation}
\Omega_2 = \frac{1}{2}\int_\Sigma d\sigma^{+-}\Big(U_+^{c,-1}\,F_{+-}^c\,U_+^c + U_+^{s,-1}\,F_{+-}^s\,U_+^s\Big) + (\text{edge commutators}).
\end{equation}
Path dependence is thus controlled separately by the integrated curvatures $F_{+-}^c$ and $F_{+-}^s$. Inserting $M = \sigma\,\openone + \mathcal{M}_\Delta^c + \mathcal{M}_\Delta^s$, one finds $F_{+-}^c$ built from $\partial_\pm \sigma$ together with the charge-sector pieces of $\Delta$, and $F_{+-}^s$ built from the non-abelian spin sector of $\Delta$. The two pieces vanish together precisely in the difermion phase, where $\Omega_1$ lives in the abelian subalgebra $\mathrm{span}\{\openone, \gamma^5\}$ and the spin sector commutes with itself.

\subsubsection{Emergent gauge field from internal phases}

Substituting $\psi = U\chi$ with the factorization~\eqref{eq:SCS-factorization} into the free kinetic term gives, to leading order in the boost rapidity $\beta$,
\begin{equation}
\mathcal{L}_0 = i\,\bar\psi\,\gamma^\mu \partial_\mu \psi = i\,\bar\chi\,\gamma^\mu \partial_\mu \chi + (\partial_\mu \theta_N)\,\bar\chi\,\gamma^\mu \chi + \tfrac{1}{2}(\partial_\mu \beta)\,\bar\chi\,\gamma^\mu \gamma^5 \chi + \mathcal{O}(\beta^2).
\label{eq:SCS-minimal-couplings}
\end{equation}
The higher-order corrections come from the non-unitarity of $U_s =  e^{- \beta\gamma^5/2}$, which satisfies $U_s^\dagger U_s = e^{-\beta\gamma^5} \neq \openone$. Note that the number current $\bar\chi\gamma^\mu\chi$ couples to $\partial_\mu\theta_N$ and the axial current $\bar\chi\gamma^\mu\gamma^5\chi$ couples to $\partial_\mu\beta$. Promoting the phase gradients to background gauge fields gives the minimal substitution
\begin{equation}
\partial_\mu \;\longrightarrow\; D_\mu = \partial_\mu - i A_\mu - i\gamma^5 B_\mu, \qquad A_\mu \equiv \partial_\mu \theta_N, \qquad B_\mu \equiv \tfrac{1}{2}\partial_\mu \beta,
\label{eq:covariantization}
\end{equation}
where the number current couples to $A_\mu$ and the axial current to $B_\mu$. With $\Delta = |\Delta|\,e^{i\beta_\Delta}$ one identifies $\beta \equiv \beta_\Delta$ up to conventions, so $\partial_\mu \beta$ encodes phase fluctuations of the pairing channel. This abelian projection of the composite connection is the emergent gauge field $\mathcal{A}_\mu$ of section~\ref{sec:local-fierz-complete} and~\cite{Haddad2024}, and it is the relevant gauge structure in the difermion phase consistent with $\Omega_1 \in \mathrm{span}\{\openone, \gamma^5\}$, in that regime.

\subsection{Wilson lines, Wilson loops as spin-charge deconfinement diagnostic}
\label{subsec:wilson-lines}

The Wilson-line/Wilson-loop framework provides the standard diagnostic for the confinement-deconfinement transition of charged states in a gauge theory (for background, see refs.~\cite{Mandelstam1968,Wilson1974,Schwinger1962}; for the slave-particle gauge-theoretic framing of SCS as a confinement problem, see refs.~\cite{MudryFradkin1994a,MudryFradkin1994b,LeeNagaosaWen2006}). Applied to the composite connection $A_\mu^{\rm dress}$, it distinguishes the two ordered phases by the scaling of Wilson-type objects: half-infinite Wilson lines are finite in the difermion phase, where the dressed quasiparticles propagate as gauge-invariant asymptotic single-particle states, while closed Wilson loops in the boost sector develop area-law scaling in the chiral phase, where the commutator self-coupling generates the $\gamma^5$-valued field strength of eq.~\eqref{eq:F-dress-chiral} and single-particle spin and charge quasiparticles cease to be asymptotic.

\subsubsection{Half-infinite Wilson lines in the difermion phase}

In gauge theories, Wilson lines provide the gauge dressing that transforms a bare fermion into a physical, gauge-invariant observable. A half-infinite Wilson line extending from the fermion's location to infinity encodes the long-range gauge field sourced by a single isolated fermion. This describes an individual asymptotic fermion dressed with its field cloud and retains the fermion's spin and statistics. Such a construction is well-defined only when single-particle states are asymptotic, i.e., only in deconfined regimes.

In the difermion phase of the present model, the dressing factor $\exp(-i\beta\gamma^5)$ provides exactly such a Wilson-line dressing. The phase $\beta$ is the Goldstone field of the broken $U(1)_\Delta$ symmetry, $\beta(x) = -\!\int_x^\infty \partial_\mu \beta(z)\,dz^\mu$ assuming asymptotic vanishing, and the diagonal action on the chiral components is
\begin{equation}
e^{-i\beta\gamma^5} = \mathrm{diag}\!\left(e^{-i\beta}, e^{+i\beta}\right),
\end{equation}
with the redefinition $\beta \to \beta/2$ absorbed into the dressing. Defining the abelian gauge field $\mathcal{A}_\mu \equiv \partial_\mu \beta$ and the chiral gauge field $\mathcal{A}^5_\mu \equiv \gamma^5 \mathcal{A}_\mu$ constructs the half-infinite Wilson line
\begin{equation}
\mathcal{W}_{(x;\infty)} = \exp\!\left(i\int_x^\infty \mathcal{A}^5_\mu(z)\,dz^\mu\right),
\label{eq:half-infinite-wilson}
\end{equation}
under which the dressed fermion $\psi'(x) = \mathcal{W}_{(x;\infty)}\,\psi(x)$ is gauge-invariant under axial $U(1)_5$ transformations $\beta(x) \to \beta(x) + \delta(x)$ with $\delta(\infty) = 0$. The dressed fermion is the asymptotic gauge-invariant single-particle state in the difermion phase.

This construction is well-defined precisely because the difermion phase is deconfined. The abelian field strength $F^5_{+-} = \partial_+ \partial_- \beta - \partial_- \partial_+ \beta$ vanishes identically (for a smooth scalar $\beta$), and the dressing factor $e^{-i\beta\gamma^5}$ contributes no Yang-Mills curvature to leading order. The half-infinite Wilson line is then finite in the asymptotic limit and provides a sensible gauge-invariant single-particle observable. A finite Wilson line $\mathcal{W}_{(x;y)}$ between two points represents two dressed fermions bound by the residual gauge structure between them. In the difermion phase this binding is loose because the connection is flat, consistent with the SCS spectrum of asymptotic spin and charge quasiparticles. A thorough treatment of this perspective on spin-charge separation can be found in~\cite{Haddad2024}.

\subsubsection{Closed Wilson loops and area law in the chiral phase}
The chiral phase has the opposite structure. The Clifford-support diagnostic of section~\ref{sec:local-fierz-complete} places $\Omega_1$ in the non-abelian boost subalgebra,
\begin{equation}
\Omega_1 = i\left[b(x)\gamma^0 + c(x)\gamma^1\right],
\end{equation}
with $b$ and $c$ sourced by the chemical potential and the difermion real and imaginary parts respectively, as in the analysis of section~\ref{sec:local-fierz-complete}. The composite connection takes values in $\mathrm{span}\{\gamma^0, \gamma^1\}$, and a half-infinite Wilson line in this sector is not finite. The phase $\beta$ of the difermion field is ill-defined in the chiral phase (the corresponding $U(1)_\Delta$ is unbroken), so the $\gamma^5$-sector half-infinite construction does not apply. The natural Wilson-type object in this phase is the closed Wilson loop in the boost subalgebra.

For a closed loop $C$ in spacetime bounding a region $S$, the Wilson loop is
\begin{equation}
W(C) = \mathrm{Tr}\,\mathcal{P}\!\exp\!\left(i\oint_C A_\mu^{\rm dress}\,dx^\mu\right),
\label{eq:wilson-loop-def}
\end{equation}
with path ordering $\mathcal{P}$ along $C$. The composite connection in the chiral phase is $A_\mu^{\rm dress} = -(\partial_\mu b)\gamma^0 - (\partial_\mu c)\gamma^1$ at leading Magnus order, valued in the non-abelian boost subalgebra. The non-abelian field strength is
\begin{equation}
F^{\rm dress}_{tx} = \partial_t A_x^{\rm dress} - \partial_x A_t^{\rm dress} - i[A_t^{\rm dress}, A_x^{\rm dress}].
\label{eq:F-tx-nonabelian}
\end{equation}
The derivative part vanishes by equality of mixed partials. The commutator part, computed using $[\gamma^0, \gamma^1] = 2\gamma^5$ from eq.~\eqref{eq:clifford-commutators}, gives
\begin{equation}
F^{\rm dress}_{tx} = 2i\,J(b,c)\,\gamma^5, \qquad J(b,c) \equiv \partial_x b\,\partial_t c - \partial_t b\,\partial_x c,
\label{eq:F-jacobian-form}
\end{equation}
which is the $\gamma^5$-valued field strength of eq.~\eqref{eq:F-dress-chiral} written in terms of the Jacobian of the map $(t,x) \mapsto (b,c)$. The field strength lies in the abelian one-dimensional sub-algebra generated by $\gamma^5$, even though the connection itself lies in the two-dimensional non-abelian boost sub-algebra. This is the standard feature of the non-abelian commutator: the curvature has support in the bracket of the algebra generators, here $\gamma^5$.

The non-abelian Stokes theorem~\cite{Arefeva1980,Bralic1980,DiakonovPetrov1989} expresses the path-ordered line integral as a path-ordered surface integral $W(C) = \mathrm{Tr}\,\mathcal{P}_S\exp(i\int_S U^{-1}(x_0,x)\,F(x)\,U(x_0,x)\,dS)$, where $U(x_0, x)$ is the parallel transport from a basepoint $x_0$ to $x$ along a chosen surface path. The conjugation $U^{-1}F U$ generically rotates the algebra value of $F$, since the boost-subalgebra generators do not commute with $\gamma^5$. For our geometry the rotation is controlled by the $\mathfrak{sl}(2,\mathbb{R})$ structure: $[\gamma^0, \gamma^5] = 2\gamma^1$ and $[\gamma^1, \gamma^5] = 2\gamma^0$ rotate $\gamma^5$ into the boost subalgebra. For weak dressings, or equivalently for surfaces with area small compared to the inverse-curvature scale, the conjugation is close to the identity and $U^{-1}F U \approx F$, so the path-ordered surface exponential reduces to an ordinary surface integral. In this leading-order regime the Wilson loop becomes
\begin{equation}
W(C) = \mathrm{Tr}\,\exp\!\left(i\int_S F^{\rm dress}_{tx}\,dt\,dx\right) = \mathrm{Tr}\,\exp\!\left(-2\,\Phi_5(C)\,\gamma^5\right),
\label{eq:wilson-loop-evaluated}
\end{equation}
where
\begin{equation}
\Phi_5(C) \equiv \int_S J(b,c)\,dt\,dx
\label{eq:gamma5-flux}
\end{equation}
is the integrated $\gamma^5$ flux through the surface $S$. Non-leading corrections from the non-abelian parallel transport modify the precise tension coefficient but not the area-law structure of the result. In the chiral basis where $\gamma^5 = \mathrm{diag}(+1, -1)$, the trace evaluates to
\begin{equation}
W(C) = 2\cosh\!\left(2\,\Phi_5(C)\right).
\label{eq:wilson-loop-final}
\end{equation}

The scaling of $\langle W(C)\rangle$ with the geometry of $C$ is determined by the statistics of the flux $\Phi_5(C)$ in the chiral phase. The boost-sector fields $b(x)$ and $c(x)$ are not Goldstone modes (the corresponding $U(1)_\Delta$ is unbroken in the chiral phase) and therefore have no stiffening that would suppress their fluctuations. They fluctuate with short-range correlations characterized by a correlation length $\xi$ and a variance scale $\langle J^2\rangle$, and by the symmetries of the chiral-phase vacuum the mean flux density vanishes, $\langle J(b,c)\rangle_{\rm chiral} = 0$. The Wilson-loop expectation value is conveniently computed in the Euclidean continuation $\Phi_5 \to i\Phi_5^E$, under which the Minkowski expression $W(C) = 2\cosh(2\Phi_5)$ becomes the Euclidean Wilson loop
\begin{equation}
W_E(C) = 2\cos\!\left(2\Phi_5^E(C)\right).
\label{eq:wilson-euclidean}
\end{equation}
The Euclidean continuation is the standard setting for area-law diagnostics, where the analog statement in Yang-Mills theory is the expectation value $\langle\mathrm{Tr}\,U(C)\rangle \sim \exp(-\sigma\cdot\mathrm{Area})$ of the closed Wilson loop on the Euclidean lattice~\cite{Wilson1974}.

For Gaussian fluctuations with vanishing mean, the variance of the flux scales with the surface area. Writing
\begin{equation}
\langle\Phi_5^E(C)^2\rangle = \int_S\int_S \langle J(x)J(y)\rangle\,d^2x\,d^2y,
\label{eq:phi-variance-integral}
\end{equation}
the integrand $\langle J(x)J(y)\rangle$ is supported on a region of size $\xi$ around the diagonal $x = y$. For surfaces $S$ with area much larger than $\xi^2$, the integral is dominated by the diagonal contribution and scales linearly with the area:
\begin{equation}
\langle\Phi_5^E(C)^2\rangle \approx \langle J^2\rangle\,\xi^2\,\mathrm{Area}(S).
\label{eq:phi-variance-area}
\end{equation}
The Gaussian identity $\langle\cos(kX)\rangle = \exp(-k^2\langle X^2\rangle/2)$ for $X$ Gaussian with mean zero then gives the area-law expectation value
\begin{equation}
\langle W_E(C)\rangle = 2\,\langle\cos(2\Phi_5^E)\rangle = 2\exp\!\left(-2\langle\Phi_5^E(C)^2\rangle\right) \sim \exp\!\left(-\sigma_{\rm sc}\cdot\mathrm{Area}(S)\right),
\label{eq:area-law}
\end{equation}
with spin-charge string tension
\begin{equation}
\sigma_{\rm sc} = 2\,\langle J^2\rangle\,\xi^2.
\label{eq:string-tension}
\end{equation}
This is the standard exponential-in-area decay characteristic of confinement: separating an asymptotic spinon-chargon pair to large distance $R$ costs energy proportional to $R$, with effective string tension $\sigma_{\rm sc}$.

Eq.~\eqref{eq:F-jacobian-form} makes the physical content of the area law explicit. The $\gamma^5$-valued field strength is generated by configurations in which $b(x)$ and $c(x)$ vary independently across the spacetime region, with non-zero local Jacobian $J(b,c)$. In the chiral phase, the boost-sector fields are unstiffened and their fluctuations are large; the spin-charge string tension $\sigma_{\rm sc}$ encodes the strength of these fluctuations multiplied by their correlation area. As the system is driven toward the difermion phase, the Clifford support of $\Omega_1$ rotates out of the boost sector and into the abelian $\{\openone,\gamma^5\}$ subalgebra. The boost-sector fields cease to be the relevant dressing variables, the variance $\langle J^2\rangle$ associated with them is no longer a meaningful quantity for the dressing, and $\sigma_{\rm sc} \to 0$ at the transition.

\subsubsection{The diagnostic across the phase diagram}
The two Wilson-type objects developed above give complementary scalings that interpolate continuously across the phase diagram. In the difermion phase, the half-infinite Wilson line $\mathcal{W}_{(x;\infty)}$ is finite, and dressed quasiparticles propagate as deconfined single-particle states. The closed Wilson loop reduces to a perimeter law: the abelian $\gamma^5$-sector loop is bounded by $2\cos(\oint_C \mathcal{A}^5_\mu\,dx^\mu)$, which depends on the perimeter of $C$ but not on the enclosed area. In the chiral phase, the half-infinite Wilson line in the boost sector is divergent (no asymptotic boost-sector single-particle states), and the closed Wilson loop has the area-law form of eq.~\eqref{eq:area-law}. Spin-charge quasiparticles are confined with tension $\sigma_{\rm sc} = 2\langle J^2\rangle\,\xi^2$. In the intermediate regime, $\Omega_1$ has support on the full Clifford basis, the field strength has contributions from both abelian and boost sectors, and the Wilson loop interpolates between the two scalings.

The Wilson-line/Wilson-loop diagnostic therefore provides a direct realization, internal to the present construction, of the structural reading developed in section~\ref{sec:synthesis}. The chiral-to-difermion transition is the deconfinement transition for the spin-charge degrees of freedom in the standard area-law sense of gauge theory, with the explicit tension formula~\eqref{eq:string-tension} computed from the dressing structure of the model.

\subsection{Formal geometric structure}
\label{subsec:flat-connection}

Here, we provide a brief comment on the geometric structure inherent in the dressed formulation. The dressing $U(x)$ takes values in the group generated by the Clifford algebra, and the composite connection $A_\mu^{\rm dress} = i(\partial_\mu U)U^{-1}$ is the content of a principal-bundle connection over $(1+1)$d spacetime, with the dressed spinor $\chi(x) = U^{-1}(x)\psi(x)$ and the field $\psi$ viewed in the local frame chosen by the connection. The two regime-specific representations developed above are reductions of this bundle to sub-bundles whose structure group is a proper Clifford subalgebra: the SCS factorization reduces the bundle to the abelian subgroup associated with $\mathrm{span}\{\openone, \gamma^5\}$, whereas the chiral-phase Wilson construction reduces it to the non-abelian boost subgroup; in the intermediate regime no proper reduction exists.

The two Wilson-type objects of section~\ref{subsec:wilson-lines} probe two distinct geometric features of this connection. The bulk of $(1+1)$d Minkowski space is contractible, so a flat connection on a patch carries no holonomy; the interesting topological content of the dressing lives at infinity, where the dressing field satisfies asymptotic boundary conditions. The half-infinite Wilson line detects this asymptotic structure, and its finiteness in the difermion phase reflects the asymptotic vanishing of the Goldstone $\beta$. The closed Wilson loop instead probes non-flatness of the connection inside the loop, and its area law in the chiral phase is the statement that $F^{\rm dress}_{+-} \neq 0$ at the relevant order, with the $\gamma^5$-valued curvature of eq.~\eqref{eq:F-jacobian-form} providing the integrated flux.

\section{Discussion}
\label{sec:synthesis}

The technical results that we have developed connect to a longstanding proposal in non-abelian gauge theory. The chiral and difermion phases of the model use complementary variables. Quantities that are classical and well-defined in one phase are strongly fluctuating or suppressed in the other, and the dressing $U(x)$ is the transformation that interpolates between the two variable sets in the manner familiar from $(1+1)$d bosonization~\cite{Coleman1975,Mandelstam1975Bosonization,Witten1984Bosonization}. In this section we discuss the relation of this structure to the CFN proposal of Niemi et al. that SCS may underlie color confinement. We will discuss the sense in which our present model realizes that proposal and describe a mapping onto the structural regions of the QCD phase diagram.

\subsection{The Niemi proposal for spin-charge separation as a confinement mechanism}

In this section we will discuss two distinct correspondences with non-abelian gauge theory. The first is a confinement-mechanism analogy with pure Yang--Mills theory: following the proposal of Niemi and collaborators, the spin-charge-separated phase plays the role of the confining vacuum. The second is a phase-diagram analogy with finite-density QCD, in which the three regimes of the model map onto three structural regions of the QCD phase diagram (section~\ref{subsec:qcd-phase-diagram}). The two are related interpretations of the same chiral-to-difermion transition, but they belong to different physical settings; one concerning the Yang--Mills vacuum, the other quark matter at finite density. We begin with the Yang--Mills mechanism.

The proposal of~\cite{NiemiSCSConfinement2005,FaddeevNiemi2007,NiemiWalet2005} is that the confining phase of four-dimensional SU(2) Yang-Mills theory is precisely the SCS phase of the gluon. In the Cho-Faddeev-Niemi decomposition, the gluon is reorganized into a spin field $\hat{n}(x)$ and two charged scalars. At high energies these constituents are tightly bound, and the Yang-Mills Lagrangian describes asymptotically free massless gauge bosons. In the low-energy domain, the charged scalars condense and the spin-charge components decouple. The resulting low-energy effective action takes the form of a coupled O(3) and Grassmannian sigma model, which supports knotted closed string solitons. The confining vacuum is then structurally analogous to a high-$T_c$ superconductor: condensation of charged scalars produces an effect that has been called a dual Meissner effect~\cite{Mandelstam1976Vortices,tHooft1981Topological}, with chromoelectric flux squeezed into the flux tubes that bind quarks.

That proposal has driven significant subsequent work but not a concrete realization, in the strict sense. As discussed in section~\ref{sec:introduction}, three obstacles remain: projective Lorentz action in the O(3) sector, absence of a derivation from the path integral, and no explicit construction of the spin-charge-separated fermion. The present model exhibits these features explicitly, albeit in the more limited setting of $(1+1)$d where the obstacles are absent.

\subsubsection{Two approaches to the Niemi correspondence}

The correspondence between the present model and the Niemi picture proceeds along two complementary lines:

\paragraph{SCS as deconfinement of spin and charge.}
As we have discussed, in the chiral phase ($\langle\sigma\rangle \neq 0$, $\langle\Delta\rangle = 0$), the emergent gauge field $\mathcal{A}_\mu$ of section~\ref{sec:local-fierz-complete} has large field strength $\mathcal{E}$, the spinon and chargon components of the dressed fermion are tightly bound, and there is no spin-charge separation. In the difermion phase ($\langle\sigma\rangle = 0$, $\langle\Delta\rangle \neq 0$), by contrast, the difermion phase $\beta$ becomes a stiff Goldstone, $\mathcal{E} \to 0$, the dressed components $\phi_\pm$ propagate independently, and SCS is realized. The intermediate regime interpolates between the two limits. The chiral-to-difermion transition is therefore a deconfinement transition for the spin and charge degrees of freedom, with $\mathcal{A}_\mu$ playing the role of the binding field.

The structural correspondence with the Niemi picture is direct. In Niemi's proposal, the YM confining phase is the SCS phase, and the IR weakening of the spin-charge binding interaction is the mechanism for the separation. In the present model, the SCS phase is the difermion phase, and the IR weakening of $\mathcal{A}_\mu$ is the mechanism. The condensation of a charged scalar (the difermion $\Delta$, which carries fermion number 2) plays the role that Niemi assigns to the charged-scalar condensation of the CFN decomposition.

\paragraph{Emergence of non-abelian gauge structure away from the SCS regime.}
The composite connection $A_\mu^{\rm dress}$ in eq.~\eqref{eq:A-dress-expanded} is abelian in the difermion phase, where $\Omega_1 \in \mathrm{span}\{\openone, \gamma^5\}$. Outside this regime, the commutator $[\Omega_1, \partial_\mu \Omega_1]$ is non-zero and the connection acquires Yang-Mills-like corrections valued in the $\mathfrak{sl}(2,\mathbb{R})$ algebra generated by the Clifford commutators of eq.~\eqref{eq:clifford-commutators}. The field strength $F^{\rm dress}_{tx}$ in the chiral phase contains a $\gamma^5$-valued contribution from the commutator self-coupling, as in eq.~\eqref{eq:F-dress-chiral}. The intermediate regime exhibits the full $\mathfrak{sl}(2,\mathbb{R})$ structure with all generators active.

The second perspective complements the first. The first focuses on the three regimes for confined/intermediate/deconfined phases of the spin-charge structure. The second focuses on the non-abelian character of the composite gauge connection: non-abelian in the chiral phase, most non-abelian in the intermediate regime, and abelian in the difermion phase. The two organizations are compatible: in the difermion phase, both threads agree, since SCS is realized (thread one) and the gauge connection is abelian (thread two). In the confining vacuum as the Niemi proposal describes it, both features coincide: it is a dual superconductor with abelian-dominated low-energy structure, and SCS is realized. The difermion phase of the present model exhibits the same combination of structural features. The chiral phase has the opposite features: no SCS (thread one), non-abelian $A_\mu^{\rm dress}$ (thread two). The intermediate regime is where both features change character. Some of the structural correspondences are summarized in Table~\ref{tab:niemi-difermion}.
\begin{table}[h]
\centering
\renewcommand{\arraystretch}{1.2}
\begin{tabularx}{\textwidth}{@{}>{\raggedright\arraybackslash}p{0.265\textwidth}>{\raggedright\arraybackslash}p{0.30\textwidth}>{\raggedright\arraybackslash}X@{}}
\toprule
Feature & Niemi-confining YM & Present difermion phase \\
\midrule
Charged scalar condensate & CFN-charged scalars condense & $\langle\Delta\rangle \neq 0$ (Cooper-pair condensate) \\
Spin-charge separation & Realized in IR & Realized in IR for $\mathcal{E}$ \\
Phase of charged condensate & Stiff Goldstone & Stiff $\beta = \arg\Delta$ \\
Gauge binding strength & Weak in bulk (dual Meissner) & $\mathcal{E} \to 0$ in IR \\
Symmetry signature & Projective Lorentz action & Linear $\mathrm{SL}(2,\mathbb{R})_1 \times \mathrm{SL}(2,\mathbb{R})_2$ doubling \\
Low-energy gauge structure & Abelian-dominated & Abelian (strict, in this regime) \\
Discrete symmetry & not identified & $\mathcal{PT}$-symmetry breaking~\cite{Haddad2024} \\
\bottomrule
\end{tabularx}
\caption{Structural correspondences between the Niemi-confining Yang--Mills picture and the difermion phase of the present model.}
\label{tab:niemi-difermion}
\end{table}

\subsubsection{Three regimes and the QCD phase diagram}
\label{subsec:qcd-phase-diagram}

Combining the two threads gives a richer mapping than either alone. The three regimes of the present model correspond structurally to three regions of the QCD phase diagram, as summarized in Table~\ref{tab:qcd-regimes}, which extends the vacuum-level comparison of Table~\ref{tab:niemi-difermion} from the confining vacuum to the full set of regimes.
\begin{table}[h]
\centering
\renewcommand{\arraystretch}{1.3}
\begin{tabularx}{\textwidth}{@{}>{\raggedright\arraybackslash}p{0.24\textwidth}>{\raggedright\arraybackslash}p{0.27\textwidth}>{\raggedright\arraybackslash}X@{}}
\toprule
Present model & Structural QCD analog & Reasoning \\
\midrule
Chiral phase ($\langle\sigma\rangle \neq 0$) & Quark--gluon-plasma-like & Elementary fermion a sharp, massive excitation (matter deconfined); spin and charge bound, $\mathcal{E}$ large (no SCS); non-abelian $A_\mu^{\rm dress}$. \\
Intermediate regime & Critical/crossover region & Both $\sigma$ and $\Delta$ heavy; no sharp quasiparticle; full $\mathfrak{sl}(2,\mathbb{R})$ structure; partial SCS. \\
Difermion~phase ($\langle\Delta\rangle \neq 0$) & Hadronic (confining) vacuum & Elementary fermion paired and fractionalized (matter confined); $\mathcal{E}\to 0$, spinon and chargon free (SCS realized); abelian low-energy gauge sector. \\
\bottomrule
\end{tabularx}
\caption{Structural correspondence between the three regimes of the present model and three regions of the QCD phase diagram, along the confinement and spin-charge axes.}
\label{tab:qcd-regimes}
\end{table}
The organizing fact behind this mapping is that confinement of the elementary fermion and separation of its spin and charge run in opposite directions across the phase diagram.\footnote{The terms ``confinement'' and ``deconfinement'' are used here in the structural sense of the present construction---binding versus liberation of the elementary fermion, and binding versus separation of its spin and charge. They are distinct from the thermal color-deconfinement of QCD, whose order parameter is the Polyakov loop and whose center-symmetry structure has no analog in the present $(1+1)$d model.} In the chiral phase the elementary fermion is a sharp, massive excitation; the natural degree of freedom (spin and charge remain tightly bound). The fermion is therefore deconfined and spin-charge is confined, which places the phase on the quark--gluon-plasma-like side: the regime in which the elementary quanta propagate freely and no spin-charge structure has emerged. In the difermion phase the situation is reversed. The charged condensate $\langle\Delta\rangle$ pairs the elementary fermion and removes it from the spectrum as a sharp excitation, while $\mathcal{E}\to 0$ lets the spinon and chargon propagate independently. The fermion is confined and spin-charge is deconfined, placing the phase on the hadronic, confining side. In the intermediate regime, both condensates are heavy and quasiparticles are not well-defined.

This anti-correlation is the structural content of the Niemi proposal: the confining (hadronic) analog is the spin-charge-separated phase (here the difermion phase) rather than the phase in which the elementary fermion is the good excitation. In QCD the hadronic vacuum confines quarks while the gluon, in the Niemi picture, is spin-charge separated; mapping quarks to the elementary fermion and the gluon spin-charge structure to ours, the difermion phase reproduces precisely this pairing of confined matter with separated spin and charge.

The correspondence holds along the confinement and spin-charge axes but \emph{not} along the chiral-condensate axis. The difermion phase has $\langle\sigma\rangle = 0$ and is chirally restored, whereas the hadronic vacuum it is matched to has $\langle\bar{q}q\rangle \neq 0$ and is chirally broken. We therefore do not identify $\sigma$ with the QCD chiral condensate; in the present mapping $\sigma$ marks the phase in which the elementary fermion is gapped but unpaired, which is the deconfined side. The naive identification: chiral phase as hadronic, difermion phase as deconfined; is the one that follows from reading $\sigma$ as $\langle\bar{q}q\rangle$, and it is exactly the identification that the confinement and spin-charge criterion reverses.

\section{Conclusion}
\label{sec:conclusion}

We have developed a self-consistent dressing framework $\psi(x) = U(x)\chi(x)$ for $(1+1)$-dimensional Dirac fermions with Fierz-complete four-fermion interactions. The dressing matrix $U(x)$ is built from the condensate background, and the composite connection $A_\mu^{\rm dress} = i(\partial_\mu U)U^{-1}$ encodes the obstruction to local trivialization of the Dirac operator. The framework extends the spin-charge separation construction of~\cite{Haddad2024} from the strong-pairing regime to the full Fierz-complete phase diagram, and displays three previously distinct non-perturbative constructions (SCS, Wilson-line dressing, flat-connection holonomy) as regime limits of the single dressing $U(x)$. The composite connection acquires Yang--Mills-like corrections at next-to-leading order in the curvature, with structure constants in the $\mathfrak{sl}(2,\mathbb{R})$ algebra of Clifford commutators. The discussion of section~\ref{sec:synthesis} places these results in the context of the Faddeev-Niemi proposal. The condensation of a charged scalar, the onset of spin-charge separation, and the change in non-abelian character of the composite connection are features of the phase transition which we have shown are also features of fermions in Fierz complete theories. Our work points to several potential future calculations.

First, the explicit calculation of the corrected gauge field $A_\mu^{\rm dress}$ at $\mathcal{O}(F_{+-})$ in section~\ref{sec:local-fierz-complete} establishes that the non-abelian portion is most pronounced in the intermediate regime. A natural next step is the computation of the effective dynamics of this non-abelian sector. For instance, the leading (quadratic) part of these dynamics should be computed where we suspect that integrating out the dressed fermions generates an induced Maxwell term for the composite connection with coefficient fixed by the same one-loop scale as the wave-function renormalization. The eqs.~\eqref{eq:intermediate-noabelian} for the leading non-abelian correction in the intermediate regime depend on cross-couplings between the real-coefficient and imaginary-coefficient blocks of $\Omega_1$. The structure of these cross-couplings, and their consequences for low-energy excitations, has not been worked out.

Second, the QCD-analog mapping above is not dynamical. A natural question is whether any analog of confinement dynamics, such as a linear potential between probe charges or a string tension, can be identified in the present model. The (1+1)d setting precludes flux tubes, but the analog of confining static potentials between heavy probe fermions in the difermion phase, if any, may be accessible through Wilson-loop calculations along the lines of section~\ref{sec:interpretations}.

Third, the non-abelian corrections we have computed in $A_\mu^{\rm dress}$ take values in the $\mathfrak{sl}(2,\mathbb{R})$ algebra of Clifford commutators. This is the symmetry algebra of the in-medium dressed-spinor representation, but it is not the gauge algebra of any familiar gauge theory. The interpretation of this algebra as a hidden symmetry of the four-fermion model and its possible relation to higher-spin or asymptotic-symmetry structures could be a fruitful future investigation.

Finally, the framework developed here, together with the construction of~\cite{Haddad2024}, provides a setting in which several conjectural ideas about strongly-correlated relativistic fermion physics admit explicit and computable realizations. The $(1+1)$d setting is of course limited, but some of the structural connections that emerge are not necessarily specific to one spatial dimension and one might expect some to carry over to higher-dimensional analogs of the same construction.

\acknowledgments
L. H. acknowledges that this work was partly carried out through the Colorado School of Mines Physics Senior Design program, under
which the student co-authors contributed
calculations, discussions, and independent verification of intermediate results
throughout the project's development. It has been a privilege to mentor such
talented undergraduates, and the program is gratefully acknowledged for making
this collaboration possible. The authors also thank the Department of Physics at
Colorado School of Mines for its support during the writing of this manuscript.

\appendix

\section{Conventions for Dirac fermions in $(1+1)$d }
\label{DiracConventions}

We review here some of the fundamentals of Dirac spinors in $(1+1)$ dimensions (for background, see Refs.~\cite{Thirring1958,Coleman1975,GrossNeveu1974}). This is by no means an exhaustive introduction but simply a few relevant basics. Beginning with the most basic structure that holds regardless of which particular interactions are involved, assuming finite mass $m$ and chemical potential $\mu$, we first review a standard representation in which the gamma matrices are real. We will then look at the same structure using a complex representation. With the time-like signature $g^{\mu \nu}= \mathrm{diag}\left( 1 , - 1 \right)$ we take the $2 \times 2$ gamma matrices to be defined in terms of the Pauli matrices as
\begin{eqnarray}
\gamma^0 = \sigma_1 \; , \;\;  \gamma^1 = - i  \sigma_2 \;, \;\; \gamma^5 =   \gamma^0   \gamma^1  = \sigma_3 \, ,  \label{myalgebra}
\end{eqnarray}
which one finds satisfy the Dirac algebra  
\begin{eqnarray}
\left\{ \gamma^\mu , \, \gamma^\nu \right\} = 2 g^{\mu \nu}\, . 
\end{eqnarray}
In addition, we will need the charge conjugation operation which is defined as $\psi^C \equiv \gamma^5 \psi^*$, where $\psi^C$ is the charge-conjugated spinor wavefunction determined by the condition $C \gamma^\mu C^{-1} = - (-1)^\mu (\gamma^\mu )^T$. We choose the sign convention on the right hand side appropriate for our choice of real gamma matrices above. Hence, taking $C = \gamma^5 = \sigma_3$ satisfies this condition. The Dirac adjoint is defined by $\bar{\psi} = \psi^\dagger \gamma^0$. The Lagrangian density for free relativistic fermions with finite mass and chemical potential is then
\begin{eqnarray}
\hspace{-1pc} \mathcal{L}  &=&  \bar{\psi}\left( i \gamma^\mu  \partial_\mu - m  + \mu \gamma^0  \right)   \psi   \, . 
\end{eqnarray}
Varying $\mathcal{L}$ with respect to $\bar{\psi}$ gives the Dirac equation for the two-dimensional spinor $\psi = (\psi_1, \,  \psi_2)^T$:
\begin{eqnarray}
i (\partial_t - \partial_x ) \psi_1  - m \psi_2 + \mu \psi_1 &=& 0  \, ,\\
i (\partial_t + \partial_x ) \psi_2  - m \psi_1 + \mu \psi_2 &=& 0 \, , 
\end{eqnarray}
which has plane-wave solutions
\begin{eqnarray}
\psi_\pm(x, t) =   e^{i (px  -  E_\pm t)}  \left( \begin{array}{l}
 \sqrt{ \frac{ E_\pm  + \mu - p  }{m} }  \\
           \sqrt{ \frac{m}{ E_\pm  + \mu - p  }}       \end{array} \right) \, . 
\end{eqnarray}
Here, the sign in the subscript denotes positive and negative energy solutions. The associated momentum space Dirac operator for finite mass and chemical potential is
\begin{eqnarray}
\mathcal{D}_\pm  &=& \left( \begin{array}{ll}
  E + \mu + p      &   \;\;\;\; \mp m    \vspace{0pc}\\
\; \;  \mp m  &       E + \mu - p          \end{array} \right) \, , 
\end{eqnarray} 
which determines positive and negative energy solutions to the Dirac equation
\begin{eqnarray}
\mathcal{D}_\pm  \,  \psi = 0 \, . \label{DopEq}
\end{eqnarray} 
The Dirac operator can be reparametrized to the hyperbolic form 
\begin{eqnarray}
\mathcal{D}_\pm  &=& \left( \begin{array}{ll}
 \; e^\eta      &   \;  \mp 1   \vspace{0pc}\\
  \mp 1 &     \;    e^{-\eta}       \end{array} \right) \, , 
\end{eqnarray} 
where $\mathrm{cosh}\, \eta \equiv (E + \mu)/m$, $\mathrm{sinh} \, \eta \equiv  p/m$, and $\mathrm{tanh} \, \eta = p/(E+ \mu)$. Here, $\eta$ is the usual rapidity which one encounters with generators of $(1+1)$d Lorentz transformations $\Lambda$
\begin{eqnarray}
\Lambda  &=& \left( \begin{array}{ll}
 \; e^\eta      &   \;   0   \vspace{0pc}\\
   \; 0  &      e^{-\eta}       \end{array} \right)  = e^{\eta \gamma^5} \, .
\end{eqnarray} 
Thus, the relativistic limit of the in-medium (finite $\mu$, small $m$, large Fermi surface) Dirac operator amounts to an ordinary Lorentz transformation. Solutions of Eq.~\eqref{DopEq} have the form 
\begin{eqnarray}
\psi_\pm   &=& \phi_\pm  \left( \begin{array}{l}
  e^{- \eta/2}    \\
  \, \pm  e^{\eta/2}     \end{array} \right) \, .  \label{SinglePartState}
\end{eqnarray} 
Alternatively, one may choose different real parametrization wherein $\mathrm{cos}\, \phi \equiv m/(E + \mu)$, $\mathrm{sin} \, \phi \equiv  p/(E + \mu)$, and $\mathrm{tan} \, \phi  = p/m$, which gives 
\begin{eqnarray}
\mathcal{D}_\pm &=& \left( \begin{array}{ll}
 1 + \mathrm{sin} \, \phi      &   \;  \mp  \mathrm{cos} \, \phi  \vspace{0pc}\\
  \mp  \mathrm{cos} \, \phi &     \;     1 - \mathrm{sin} \, \phi        \end{array} \right) \, , 
\end{eqnarray} 
with solutions of the form 
\begin{eqnarray}
\psi_\pm   &=& \theta_\pm   \left( \begin{array}{l}
  \mathrm{cos} \, \phi/2   \\
   \pm  \mathrm{sin} \, \phi/2    \end{array} \right) \, . 
  \end{eqnarray}

One may also obtain complex solutions by choosing the reduced gamma matrices to be
\begin{eqnarray}
\gamma^0 =   \sigma_1 \; , \;\;  \gamma^1 = i  \sigma_3 \;, \;\; \gamma^5 =   \gamma^0   \gamma^1  =   \sigma_2 \, , \;\; C = - i \gamma^1 = \sigma_3 \,  ,  \label{altalgebra}
\end{eqnarray}
where, again we have that $C$ satisfies $C \gamma^\mu C^{-1} = - (-1)^\mu (\gamma^\mu )^T$. With this choice the Dirac equation is 
\begin{eqnarray}
     (i \partial_t + \mu) \psi_1  - (  \partial_x  + m  ) \psi_2  &=& 0  \, ,\\
  (i \partial_t + \mu) \psi_2   + (\partial_x  - m ) \psi_1 &=& 0 \, , 
\end{eqnarray}
with Dirac Hamiltonian
\begin{eqnarray}
\hat{H}_D = i \gamma^0 \gamma^1 \partial_x + m \gamma^0 - \mu \, . 
\end{eqnarray}
Plane-wave solutions are
\begin{eqnarray}
\psi_\pm(x, t) =   e^{i (px  -  E_\pm t)}  \left( \begin{array}{l}
 \sqrt{ \frac{ E_\pm  + \mu   }{m - i p} }  \\
           \sqrt{ \frac{m - i p }{ E_\pm  + \mu  }}       \end{array} \right) \, , 
\end{eqnarray}
and have dispersion $E_\pm + \mu = \pm \sqrt{ p^2 + m^2}$. The momentum space Dirac operator takes the form
\begin{eqnarray}
\mathcal{D}_\pm  &=& \left( \begin{array}{ll}
 ( i  p - m)      &   \;\;    \;\mp  ( E + \mu) \vspace{0pc}\\
 \;\; \mp (E + \mu)  &     \;\; - ( i  p + m  )          \end{array} \right) \, , 
\end{eqnarray} 
with positive and negative energy branches satisfying $\mathcal{D}_\pm  \psi_\pm  = 0$. A trigonometric reparametrization leads to 
\begin{eqnarray}
\mathcal{D}_\pm  &=& \left( \begin{array}{ll}
 \; e^{-i \phi}      &   \;  \mp  1   \vspace{0pc}\\
  \mp 1 &     \;    e^{i \phi }       \end{array} \right) \, , 
\end{eqnarray} 
where $\mathrm{cos}(\phi) \equiv  m/(E +  \mu)$, $\mathrm{sin} ( \phi ) \equiv  p/(E + \mu)$, and $\mathrm{tan}( \phi )= p/m$, and solutions now take the form 
\begin{eqnarray}
\psi_\pm  &=& \theta \left( \begin{array}{l}
  e^{ i \phi /2}    \\
  \, \pm  e^{  -  i \phi/2}     \end{array} \right) \, .  \label{SinglePartState-a}
\end{eqnarray} 
In this complex form, Lorentz transformations are now given by
\begin{eqnarray}
\Lambda  &=& \left( \begin{array}{ll}
 \; e^{- i \phi}     &   \;   0   \vspace{0pc}\\
   \; 0  &      e^{i \phi}       \end{array} \right)  = e^{-  \phi \gamma^0 \gamma^5} \, .
\end{eqnarray} 
The real and complex representations are physically equivalent. We review both because the dressed/composite framework couples the two formulations, with the additional degrees of freedom in the lifted solution coming from a complex potential in the standard framework.

\section{Mean-field analysis }
\label{meanfield1}

Our starting point is the Lagrangian density for $N$ species of relativistic fermions with finite chemical potential. In light of our analysis from the previous section, we incorporate self-interactions whose pairing fields exhaust the possible forms in $(1+1)$ dimensions. Physically, these are generated by scalar (meson) and difermion (superconducting) channels and has the form 
\begin{eqnarray}
\hspace{-1pc} \mathcal{L}  &=& \sum_{a}  \bar{\psi}_a   \left( i \gamma^\mu  \partial_\mu + \mu \gamma^0  \right)   \psi_a  + \frac{g_s}{2}  \left(  \sum_a  \bar{\psi}_a     \psi_a  \right)^2  + \frac{g_d}{2}  \sum_{a, b} \left(   \bar{\psi}_a  C   \bar{\psi}^T_b  \right)  \left(   \psi^{T}_b   C     \psi_a   \right) ,  \label{Lag1}
\end{eqnarray}
where the physical parameters $\mu$, $g_s$ and $g_d$ are the chemical potential, scalar and difermion couplings, respectively, $C$ is the charge conjugation operator and summation over flavor (species) indices $a, \, b$, is shown. We adopt an off-diagonal pairing convention for the difermion term to avoid issues when $C$ is chosen to be a symmetric matrix. Classical results are accurate when $N \gg 1$, so that fermions may acquire dynamically generated mass and dfermion expectation value through the interactions. In fact, the low-temperature mean-field limit of eq.~\eqref{Lag1} allows for several non-vanishing correlations:
\begin{eqnarray}
\rho = \langle \bar{\psi } \gamma^0 \psi \rangle \, , \;\;\; \sigma =  \langle \bar{\psi} \psi \rangle \, , \;\;\;  \Delta = \langle  \psi^T C \psi \rangle   \, , \;\;\;  \bar{\Delta}   = \langle  \bar{\psi} C \bar{\psi}^T \rangle  \, , 
\end{eqnarray}
respectively, the average density, chiral, difermion and conjugate difermion condensates. The four-fermion interactions can be decoupled by introducing the scalar and difermion auxiliary fields, $\sigma \in \mathbb{R}$ and $\Delta \in \mathbb{C}$, and by rescaling the coupling constants $g_s \to g_s/N$ and $g_d \to g_d/N$. Suppressing the species indices and summations, the Lagrangian may be expressed as 
\begin{eqnarray}
\hspace{-1pc} \mathcal{L}  &=&  \bar{\psi}   \left( i \gamma^\mu  \partial_\mu + \mu \gamma^0  \right)   \psi  - \sigma   \,  \bar{\psi}     \psi    -  \bar{\Delta } \, \psi^T  C  \psi - \frac{N}{2 g_s} \sigma^2 - \frac{N}{2 g_d} \Delta \bar{\Delta} . \label{Lag2}
\end{eqnarray}
Integrating out the fermion fields gives the effective action
\begin{eqnarray}
\hspace{-1pc} S_\mathrm{eff}  &=& - i \,  \mathrm{Tr}\,  \mathrm{ln}  \! \left( i \gamma^\mu  \partial_\mu + \mu \gamma^0 - \sigma -  \bar{\Delta} \gamma^1 \right)    - \int \!d^2x  \left(  \frac{ \sigma^2}{2g_s} + \frac{ \bar{\Delta}^2 }{ 2 g_d }   \right),  \label{Lag2-a}
\end{eqnarray}
where the trace is over coordinate, species, and spin degrees of freedom and we have taken $\bar{\Delta} \in \mathbb{R}$. The fermion determinant in the first term accounts for summing over fermion loops and the partition function can then be expressed as
\begin{eqnarray}
\mathcal{Z} = \int  \! \mathcal{D} \sigma \mathcal{D} \bar{\Delta} \; \mathrm{exp} \! \left(  i N S_\mathrm{eff}[ \sigma, \bar{\Delta }  ]     \right)
\end{eqnarray}
Uniform condensate solutions are obtained through the saddle-point equations 
 \begin{eqnarray}
 \frac{\partial S_\mathrm{eff} }{\partial \sigma} = 0  \;\;  \;    \Rightarrow \; \;   \; \mathrm{Tr} \int \frac{d^2p}{(2 \pi)^2} \frac{-i }{\slashed{p} + \mu \gamma^0 - \sigma_c - \bar{\Delta}_c \gamma^1  }  - \frac{\sigma_c}{g_s} &=& 0 \; , 
 \end{eqnarray}
 and
 \begin{eqnarray}
 \frac{\partial S_\mathrm{eff} }{\partial \bar{\Delta} }       \;\;  \;    \Rightarrow \; \;   \;        \mathrm{Tr} \int \frac{d^2p}{(2 \pi)^2} \frac{-i \,   \gamma^1}{\slashed{p} + \mu \gamma^0 - \sigma_c - \bar{\Delta}_c \gamma^1  }  - \frac{\bar{\Delta}_c}{g_d} &=& 0  \, , 
 \end{eqnarray}
 where the subscripts on the fields denote the classical solutions. Multiplying the top and bottom of each integrand by its denominator and taking the trace gives
  \begin{eqnarray}
 \int \frac{d^2p}{(2 \pi)^2} \frac{  2  \,  \sigma_c }{ (p_0 + \mu)^2 + ( p_1 + \bar{\Delta}_c )^2 + \sigma_c^2 }  &=&  \frac{\sigma_c}{g_s}  \; ,  \label{class1}
 \end{eqnarray}
 and
 \begin{eqnarray}
  \int \frac{d^2p}{(2 \pi)^2} \frac{  2  \,  \bar{\Delta}_c }{ (p_0 + \mu)^2 + ( p_1 + \bar{\Delta}_c )^2 + \sigma_c^2 }  &=&  \frac{     \bar{\Delta}_c}{g_d}  \; .  \label{class2}
 \end{eqnarray}
For a non-vanishing classical solution, the integrals can be cast into Euclidean polar form 
 \begin{eqnarray}
\frac{1}{2 \pi^2} \int_0^\infty dp' \int_0^{2 \pi} d\theta \,   \frac{ \sqrt{ p'^2 + \delta^2 - 2 \,  p' \delta \cos \theta }}{ p'^2 + \sigma_c^2} = \frac{1}{g} \, , \label{polar}
 \end{eqnarray}
 where we have incorporated the background shift $\vec{p}' \equiv \vec{p} - \vec{\delta}$, with $\vec{\delta} \equiv (\mu , \, i  \bar{\Delta}_c ) \Rightarrow  \delta = \sqrt{\mu^2 - \bar{\Delta}_c^2}$. Several features of eq.~\eqref{class1}-\eqref{polar}:
 \begin{enumerate}
 
 \item The condensates $\sigma_c$ and $\bar{\Delta}_c$ are simultaneously non-zero only when $g_s = g_d$. For general values of the couplings one or both condensates must vanish. 
 
 \item In the presence of a non-vanishing background (and $g_s \simeq g_d$) the UV limit of the integral ($p_0, \, p_1 \gg \mu , \, \bar{\Delta}_c$) reduces eq.~\eqref{class1} to the standard Gross-Neveu model, which diverges logarithmically. 
 
 \item  The angular portion of the integral in eq.~\eqref{polar} can be put into elliptic form of the second kind. For large fermion momentum the numerator in eq.~\eqref{polar} can be expanded for $p' \gg \delta$, so that the angle integrations can be performed term by term to yield
  \begin{eqnarray}
\frac{1}{2 \pi^2  }  \int_\epsilon^\kappa dp'   \frac{  p' }{ p'^2 + \sigma_c^2}       \, \left[  2 \pi  - \frac{\pi}{16} \left( \frac{\delta}{p'}\right)^2   -   \frac{15\,  \pi}{512} \left( \frac{\delta}{p'}\right)^4 - \dots     \right]      = \frac{1}{g} , \label{polar2}
 \end{eqnarray}
where UV and IR cutoffs are introduced such that $\kappa, \, \epsilon \gg \delta$. The first term in the expansion gives the usual Gross-Neveu result 
\begin{eqnarray}
   \frac{1}{ \pi }   \lim_{\kappa \to \infty }    \ln \sqrt{  \frac{  \kappa^2 + \sigma_c^2   }{  \epsilon^2 + \sigma_c^2  } } \, , \label{scalarres}
\end{eqnarray}
which alone yields the classical solution 
\begin{eqnarray}
\sigma_c^2 = \frac{\kappa^2 e^{-2 \pi/g_s} - \epsilon^2}{1 - e^{-2 \pi/g_s}} \, . 
\end{eqnarray}
For weak coupling and removing the IR cutoff we obtain the familiar Gross-Neveu result $\sigma_c  =  \kappa \, e^{-\pi/g_s}$. Fixing the left hand side, reveals the expected asymptotic freedom in $2$D. Integrating the second term in the expansion gives an additional logarithmic correction:
\begin{eqnarray}
    \frac{1}{\pi}  \lim_{\kappa \to \infty }    \ln\left(     \frac{\epsilon}{ \kappa}  \sqrt{  \frac{  \kappa^2 + \sigma_c^2}{    \epsilon^2 + \sigma_c^2  }  } \right)^{\delta^2/32\sigma_c^2}   \, .  \label{difermionres}
\end{eqnarray}
In the UV limit that we are considering here the chemical potential is large such that $\delta \gg \sigma_c$, and the combination of results in eq.~\eqref{scalarres} and \eqref{difermionres} leads to the result
\begin{eqnarray}
\bar{\Delta}_c =  \lim_{\kappa \to \infty  ,  \, \sigma_c \to 0} \sqrt{ \mu^2 + \frac{32 \pi \sigma_c^2}{g_d \ln \! \sqrt{  (1+ \sigma_c^2/\kappa^2)/(1+ \sigma_c^2/\epsilon^2)} }}    =    \sqrt{ \mu^2 - \frac{ 64 \pi  \epsilon^2}{ g_d }} \, . 
\end{eqnarray}
The lower bound for difermion condensate formation, $\mu_\mathrm{min}^d = 8 \epsilon ( \pi /g_d)^{1/2}$, is consistent with the in-medium IR cutoff being essentially the characteristic momentum at the Fermi surface edge. This result must also be consistent with our initial requirement that $\epsilon \gg \mu$. At the lower bound we then have $64 \pi/g_d \ll 1$, which says that this regime is characterized my large difermion coupling, consistent with the formation of difermion bound pairs.

\item At low chemical potential and well below the threshold for difernion formation, the theory is dominated by the IR region so that we may take $p, \,\bar{ \Delta}_c \ll \sigma_c, \, \mu$, in eq.~\eqref{class1}. To lowest order in the small quantities, eq.~\eqref{class1} reduces to
 \begin{eqnarray}
\frac{1}{2 \pi  }  \int_\epsilon^\kappa dp \frac{  2 \, p}{ \mu^2  + \sigma_c^2}      = \frac{1}{g_s} \, , \label{IR1}
 \end{eqnarray}
which gives a simple result after integrating and removing the IR cutoff:
\begin{eqnarray}
 \sigma_c = \sqrt{  \frac{\kappa^2 g_s}{ 2 \pi } - \mu^2} \, , \label{IR2}
 \end{eqnarray}
with the scalar condensate vanishing at $\mu_\mathrm{max}^s = \kappa (g_s/2 \pi)^{1/2}$. Our approach sets a lower bound $\mu_\mathrm{min}^s  > \kappa$, where $\sigma_c(\mu_\mathrm{min}^s) = \kappa \sqrt{ g_s/(2 \pi) -1 }$, which forces a minimum value for the scalar coupling $g_s$ for this region. Hence, the upper cutoff $\kappa$ is directly related to the energy of the mass gap.

 \end{enumerate}

 A complete treatment would require a renormalization-group analysis. Our preliminary treatment nevertheless reproduces the correct qualitative features for scalar and difermion condensation, e.g., the scalar condensate decreasing with $\mu$, vanishing at some critical point $\mu_\mathrm{max}^s$, with the difermion condensate increasing beyond some critical threshold $\mu_\mathrm{min}^d$. Finally, our approach shows that a coexistence region exists if $\mu_\mathrm{min}^d <   \mu_\mathrm{max}^s \Rightarrow     64 \pi/g_d < g_s/2 \pi$, which is satisfied when both couplings are large. Incorporating the constraint eq.~\eqref{class1}-\eqref{class2}, shows that couplings are pined in this region through $g_d = g_s \equiv g >  8 \sqrt{2} \,  \pi $. The exact values here are not rigorous results, but the essential behavior agrees with more elaborate treatments.

\bibliographystyle{JHEP}
\bibliography{refs_auto}
\end{document}